\documentclass[a4paper,11pt,reqno]{amsart}
          \usepackage{amssymb}
	  \usepackage{amsmath}
	  \usepackage{amsthm}
          \usepackage{amsfonts}
          \usepackage[english]{babel}
          \usepackage[utf8]{inputenc}
          \usepackage{enumitem}

\usepackage[margin=2.5cm]{geometry}

\usepackage{tikz}
\usetikzlibrary{arrows,shapes}

 \usepackage[unicode,colorlinks,plainpages=false,hyperindex=true,bookmarksnumbered=true,bookmarksopen=false,pdfpagelabels]{hyperref}
 \usepackage{caption}
 \hypersetup{urlcolor=cyan,linkcolor=blue,citecolor=red,colorlinks=true}

\newtheorem{thm}{Theorem}[section]
\newtheorem{cor}[thm]{Corollary}
\newtheorem{lem}[thm]{Lemma}
\newtheorem{prop}[thm]{Proposition}

\newtheorem{defn}[thm]{Definition}

\newtheorem{conj}[thm]{Conjecture}
\theoremstyle{definition}
\newtheorem{rem}[thm]{Remark}
\newtheorem*{conv*}{Conventions}
\newtheorem*{conte*}{Content}

\numberwithin{equation}{section}

\newcommand{\be}{\begin{equation}}
\newcommand{\ee}{\end{equation}}
\newcommand{\bea}{\begin{eqnarray}}
\newcommand{\eea}{\end{eqnarray}}
\newcommand{\ba}{\begin{aligned}}
\newcommand{\ea}{\end{aligned}}

 \usepackage{color}
\definecolor{lila}{rgb}{1,0.2,0.9}
\definecolor{purple}{rgb}{0.4,0,1}
\definecolor{darkgreen}{rgb}{0,0.5,0}

\def\cA{{\mathcal A}}                       %
\def\sing{\mathrm{sing}}                    %
\def\bT{{\mathbb T}}                        %
\def\cp{\mathbb{C}P^{n-1}}                  %
\def\loc{\mathrm{loc}}                      %
\def\reg{\mathrm{reg}}                      %
\def\ri{{\mathrm{i}}}                       %
\def\1{{\mbox{\boldmath $1$}}}              %
\def\bC{{\mathbb C}}                        %
\def\bR{{\mathbb R}}                        %
\def\bZ{{\mathbb Z}}                        %
\def\tr{\mathrm{tr}}                        %
\def\diag{\mathrm{diag}}                    %
\def\SU{\mathrm{SU}}                        %
\def\fH{{\mathfrak{H}}}                     %
\def\cH{{\mathcal H}}                       %
\def\dt {\left.\frac{d}{dt}\right|_{t=0}}   %
\def\Ad{\mathrm{Ad}}                        %
\def\U{\mathrm{U}}                          %
\def\reg{\mathrm{reg}}                      %
\def\cL{{\mathcal L}}                       %
\def\cS{{\mathcal S}}                       %
\def\su{{\mathfrak{su}}}                    %
\def\k{{\mathfrak k}}                       %
\def\cT{{\mathcal T}}                       %
\def\cQ{{\mathcal Q}}                       %
\def\fc{{\mathfrak c}}                      %
\def\cM{{\mathcal M}}                       %
\def\cP{{\mathcal P}}                       %
\def\fF{{\mathfrak F}}                      %
\begin{document}

\title{Spherical singularities  in compactified \\ Ruijsenaars--Schneider  systems
}

\maketitle

\begin{center}

L.~Feh\'er${}^{a,b}$  and H.R.~Dullin${}^c$  \\

\bigskip

${}^a$Department of Theoretical Physics, University of Szeged, Hungary\\
${}^b$HUN-REN Wigner Research Centre for Physics, Budapest, Hungary\\
e-mail: lfeher@physx.u-szeged.hu \\
${}^c$School of Mathematics and Statistics, University of Sydney, Australia\\
e-mail: holger.dullin@sydney.edu.au

\end{center}

\begin{abstract}
We investigate certain Liouville integrable systems
constructed earlier
via reduction of the quasi-Hamiltonian double of $\SU(n)$.
These systems
live on compact connected symplectic manifolds of dimension $2(n-1)$ and
 can be interpreted as
compactified trigonometric Ruijsenaars--Schneider systems.
Depending on the value of a parameter $0<y< \pi$, they arise in two drastically different forms:
in type (i) these are toric systems, while in the type (ii) cases they possess globally continuous
action variables that generate a Hamiltonian torus action (only) on a dense
open subset of the phase space.
The principal goal of the paper is to study those fibers of the action map
(alias the $\bT^{n-1}$ momentum map)
which are contained in the complement
of the domain of the densely defined torus action occurring in the type (ii) cases.
We demonstrate that \emph{all such `singular fibers'
are smooth connected isotropic submanifolds.}
We also work out a model of the fibers as quotient spaces of certain
subgroups of $\SU(n)$ with respect to an action of another subgroup.
The general results are exemplified by determining  the vertices
of the polytope filled by the action variables in the
simplest type (ii) cases that appear for any $n\geq 4$  with $\pi/(n-1) <y < \pi/(n-2)$,
and proving that the fibers
over the `singular vertices' are diffeomorphic to
 $S^3 \simeq \SU(2)$ in these cases.
In this way, our findings enrich the set of examples of Liouville
integrable systems with spherical singularities.
 \end{abstract}

 \setcounter{tocdepth}{2}

 \newpage

\tableofcontents

\newpage

\section{Introduction}\label{Sec1}

Liouville integrable Hamiltonian systems are understood for regular values of the integral map \cite{Arn}.
The study of critical values and the corresponding singular fibers is an active area of research.
When the integral map is the momentum map of an effective Hamiltonian action of the $N$-torus on
a compact connected $2N$-dimensional  symplectic manifold, Delzant's theorem \cite{D} states that the image
of the momentum map is an $N$-dimensional convex polytope with certain additional restrictions. The preimage of a regular
value in the interior of the polytope is an $N$-dimensional Liouville torus.
The singular values fill the lower-dimensional faces, and the corresponding singular fibers are simply tori of smaller dimension
down to edges with fibers $S^1$ and vertices with fibers just a point.
The existence of an $N$-torus action is very restrictive, and most integrable systems are not toric.
The typical (technically speaking `non-degenerate')  singular fibers of Liouville integrable systems
have been classified, see, e.g.,   \cite{BO,Zung}.
Relatively recently \cite{B} so-called spherical singularities have been discovered,
which posses singular fibers that are non-toral embedded submanifolds, often spheres.
 Even though they are degenerate singularities, they do seem to occur commonly
in integrable systems studied in mathematical physics,
for example in the celebrated Gelfand--Cetlin systems  \cite{BMZ,CK,CKO,NU}
introduced by Guillemin and Sternberg \cite{GS}.
Several systems exhibiting spherical singularities are related to
toric degenerations \cite{HK, NNU,NU-pol}.
For a more detailed review of general background, see  \cite{K}.

The particular systems we are going to investigate
were constructed in \cite{FK}
by reduction of a dynamical system on $\SU(n) \times \SU(n)$.
 On an open dense subset of
their  compact phase space  these systems behave like the systems obeying Delzant's theorem.
The image of the `action map'
is a convex polytope, but some vertices violate Delzant’s conditions
and we show that
some of these non-Delzant vertices are spherical singularities.
The pertinent systems have a physical
interpretation as compactified trigonometric Ruijsenaars--Schneider systems \cite{RIMS95,RS}.
This  interpretation is well known to experts on Calogero--Moser--Sutherland type integrable many-body systems,  which
have numerous applications
in physics and connections to several areas of mathematics \cite{vDVrev,RBanff}.
It is a tied to a coordinate system on the phase space that we will not use in the present work.
 For the interpretation and the history of these systems, see also \cite{FKlim} and Appendix \ref{AppA}.

We need to recall the basics of the construction of the
systems of interest from \cite{FK}.
It begins with the manifold
\be
P:= \SU(n) \times \SU(n) = \{(A,B)\},
\label{I.1}\ee
and the map $\mu: P \to \SU(n)$ given by the group commutator
\be
\mu(A,B):= A B A^{-1} B^{-1}.
\label{I.2}\ee
This is an $\SU(n)$ equivariant map if we let $\SU(n)$ act on itself by conjugations and on $P$
by the diagonal conjugations: $\eta \cdot (A,B):= (\eta A\eta^{-1}, \eta B \eta^{-1})$, $\forall \eta\in \SU(n)$.
 The center of $\SU(n)$ is the set of scalar matrices $z \1_n$  with $z$ any $n$-th root of unity,
which we identify with $\bZ_n$.
Notice that it is the adjoint group $\SU(n)/\bZ_n$ that acts effectively.
As a result of the equivariance, the inverse image $\mu^{-1}(\mu_0)$ is stable with respect to the isotropy group
$\SU(n)_{\mu_0}$ for any  $\mu_0\in \SU(n)$.
Let us call $\mu_0$  a `strongly regular'  value of the map $\mu$ if $\SU(n)_{\mu_0}/\bZ_n$ acts freely on $\mu^{-1}(\mu_0)$.
For strongly regular $\mu_0$, $\mu^{-1}(\mu_0)$ is an embedded submanifold of $P$ and the quotient space
\be
P(\mu_0):= \mu^{-1}(\mu_0)/ \SU(n)_{\mu_0}
\label{I.3}\ee
possesses a unique smooth\footnote{Throughout  the paper, `smooth' means $C^\infty$.
 The terminology `strongly regular' is justified since in our setting `strongly regular value' implies regular value
 in the usual sense \cite{AMM}.}
 manifold structure for which the projection $\pi_0: \mu^{-1}(\mu_0) \to P(\mu_0)$ is a smooth submersion.
  This quotient space is always compact, connected and it comes equipped with a symplectic form,  denoted
 $\omega_{\mu_0}$.
 Moreover,  the conjugation invariant smooth functions of either of the matrices $A$ or $B$ descend
 to Hamiltonians in involution on the symplectic manifold $P(\mu_0)$.
 Choosing one of these matrices, generically one obtains  $(n-1)$ independent Hamiltonians.
 The origin of the symplectic form will be illuminated in Section \ref{Sec2} below.

In order to obtain a Liouville integrable system $P(\mu_0)$ should be a smooth
manifold and as small as possible. This can be achieved by taking $\mu_0$ to be of the form
\be
\mu_0(y):=\diag\left(e^{2\ri y}, \ldots, e^{2\ri y}, e^{-2(n-1) \ri y}\right)
\label{I.5}\ee
with  a suitable real parameter $0\leq y \leq\pi$.
 We call $y$ \emph{admissible} if $\mu_0(y)$ is strongly regular   and  is not
a multiple of the unit matrix $\1_n$.
This is always the case except for the finite set of excluded  values  \cite{FK} for which
\be
e^{2m \ri y} = 1\quad\hbox{for some} \quad
m=1,\dots, n.
\label{I.6}\ee
For admissible $y$, $\dim(P(\mu_0(y))) = 2(n-1)$.

 We choose to consider the integrable system on $P(\mu_0(y))$
obtained by using the conjugation invariant globally smooth functions of $B$
 to define the commuting Hamiltonians.
For any admissible $y$, action variables for this system are encapsulated by an `action map'
\be
(\hat \beta_1, \dots, \hat \beta_{n-1}): P(\mu_0(y)) \to \bR^{n-1}
\label{I.7}\ee
constructed in \cite{FK}  using essentially the eigenvalues of $B$.
The $\hat \beta_i$ are smooth functions either on the whole of $P(\mu_0(y))$ or
on a dense open submanifold thereof, and are always globally continuous.
They form the momentum map for an effective action of the torus $\bT^{n-1}$,
at least on a dense open submanifold. Whether this dense open submanifold
is the whole of $P(\mu_0(y))$ or a proper subset, depends on the value
of $y$. This dichotomy classifies the  admissible  cases  (values of $y$)  as type (i) and type (ii).
In the type (i) cases one has a globally defined $\bT^{n-1}$ action,
and $P(\mu_0(y))$ is a toric manifold, actually isomorphic with the standard toric manifold $\cp$ (equipped with a $y$-dependent
multiple of the Fubini-Study symplectic form).

The following proposition is a reformulation and slight completion
  of a result from \cite{FK}.
  For this, let us recall that the Farey sequence $\fF_n$ of order $n$ is the set of all
 fractions $0 \le a/b \le 1$ with $\gcd(a,b) = 1$ and $b \le n$ written in increasing order (see, for example, \cite{HW,R})
A pair of two consecutive elements in $\fF_n$ are referred to as  Farey neighbours.
We say that a subinterval of $[0,\pi]$ is of type (i) if it entirely consists of type (i) values of $y$,
and similarly for type (ii).
 Equivalently, we  refer to type (i) and type (ii) subintervals of $[0,1]$
using $y/\pi$ as the parameter.

\begin{prop}\label{thmI1}
The Farey sequence $\fF_n$ gives the excluded values for $y/\pi$.
For $1 \leq \kappa  \leq  (n-1)$ coprime to $n$,
the two open intervals between  $\kappa/n$ and its
neighbours in $\fF_n$ are of type (i),
and all other intervals between neighbours in $\fF_n$ are of type (ii).
\end{prop}

 \begin{rem}
For  $1 \leq \kappa \leq (n-1)$ coprime to $n$, the Farey neighbours of $\kappa/n$ are $a/b$ and $c/d$, where
$a,b,c,d$ with the conditions $1 \leq b,d \leq (n-1)$ are the unique integer solutions of the equations
\be
b \kappa - a n = 1
\quad \hbox{and}\quad n c - \kappa d = 1.
\label{HardyI1}\ee
For small $n$, the type (i) and type (ii) intervals,
are illustrated by Figure \ref{Fig2}.
There and below, we restrict to $0 \leq y \leq \pi/2$, since the systems associated to $y$ and $(\pi -y)$ are isomorphic \cite{FK}.
It can be shown  that the number of type (i) and type (ii) intervals grows as $O(n)$ and as $O(n^2)$, respectively,
so  type (ii) is rather the rule than the exception.
These statements follow easily from standard number theoretic results \cite{HW}.
See also Appendix \ref{newAppB}.
 \end{rem}

\begin{figure}
\includegraphics{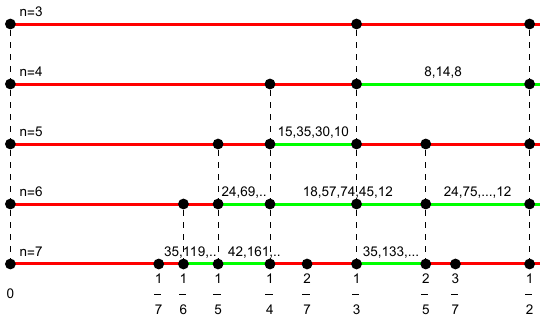}
\caption{Intervals in the parameter $y/\pi$ of type (i) (red) and type (ii) (green) for $n=3,\dots, 7$. The black dots mark
excluded values of $y/\pi$. The numbers over the  type (ii) intervals list the `face
vector' (\# vertices, \# edges, ..., \# facets) of the corresponding polytope $\cA_y$, with more details contained in Appendix \ref{AppC}.}
\label{Fig2}
\end{figure}

The image of the momentum map $\hat \beta$ \eqref{I.7} is a convex polytope  of dimension $(n-1)$ for every admissible parameter $y$.
This polytope, denoted $\cA_y$, is defined by $2n$ inequalities \eqref{T24} in an $(n-1)$-dimensional hyperplane  in
$\bR^n$ and is contained in the  simplex $\cA$  \eqref{T8}.
It is known \cite{FK} that the fiber of $\hat \beta$ over
any interior point
 of $\cA_y$ is a Liouville torus $\bT^{n-1}$, and a lower dimensional torus for those points of the boundary
 $\partial\cA_y$ that lie in the interior $\cA^\reg \subset \cA$ \eqref{T8}
(in particular, a fixed point for any vertex belonging to $\cA^\reg$).
In the type (i) cases the polytope $\cA_y$ is an $(n-1)$-dimensional simplex in $\cA^\reg$, which is
the simplest example of a Delzant polytope \cite{A,D}.
The differentiability of $\hat \beta$  is lost at $\hat \beta^{-1}(\partial \cA_y \cap \partial \cA)$,
 which gives the non-empty complement
of the domain of the torus action in every type (ii) case.

The goal of the present paper is to initiate the study of the
momentum polytopes $\cA_y$ for the type (ii) cases  and  to uncover the structure of the fibers
of the momentum map $\hat \beta$  over the `singular points' that form $\partial \cA_y \cap \partial \cA$.
 With full details contained in the text, our main result can be formulated as follows.

\begin{thm}\label{thmmain}
For every type (ii) parameter $y$,
each fiber of the momentum map $\hat \beta$ is a  smooth connected isotropic  submanifold of $P(\mu_0(y))$.
Specifically, for every $\xi \in \cA_y$, the fiber $\hat \beta^{-1}(\xi)$ is diffeomorphic to a smooth quotient
space $\SU(n)_{\delta(\xi)}/\SU(n)_\xi^{u_0}$,
where $\SU(n)_{\delta(\xi)}$  is the isotropy group of $\delta(\xi)$ \eqref{T9} with respect to conjugations, and
$\SU(n)_\xi^{u_0} < \SU(n)_{\delta(\xi)}$ is a subgroup
defined in Section \ref{Sec3} together with its action \eqref{act33} on $\SU(n)_{\delta(\xi)}$.
\end{thm}

It will be shown that the theorem recovers known results for $\xi \in \cA_y \cap \cA^\reg$, i.e. that in these  cases
the fibers are tori of various dimensions,
and is genuinely new for $\xi \in
\partial \cA_y \cap \partial \cA$. The interesting structure of the fibers over the `singular boundary'
 $\partial \cA_y \cap \partial \cA$ will
be illustrated by examples.

\begin{prop}\label{propmain}
For $\SU(n)$ with any $n\geq 4$,  the smallest type (ii) parameters vary in the interval $ \pi/(n-1) <y < \pi/(n-2)$.
In these cases  $\partial \cA_y \cap \partial \cA$ contains only vertices, $n (n-3)$ of them, and
 the fiber $\hat \beta^{-1}(\xi)$ is diffeomorphic to $S^3 \simeq \SU(2)$ for each of these vertices.
\end{prop}

\begin{conte*}
Section \ref{Sec2} is a summary of the necessary  prerequisites from \cite{FK}.
The heart of the paper is Section \ref{Sec31}, where we derive the above claimed properties of the fibers of the
momentum map.  The detailed form of the main result is given by Theorems \ref{thm314} and \ref{thm316}, together
with Definition \ref{defn36} and Lemma \ref{lem310}.
The construction and the results given in this section are valid for any $\xi \in \cA_y$.
It is explained in Section \ref{Sec32} how they reproduce the known structure of the fibers for $\xi$
belonging to the interior $\cA^\reg$  of the simplex $\cA$ \eqref{T8}.
Sections \ref{Sec4} and \ref{Sec5} are devoted to examples.
In Section \ref{Sec4}, we describe the momentum polytopes and the fibers corresponding to its singular vertices
(i.e., vertices in $\partial \cA$)
in the lowest dimension where type (ii) cases occur. Our new results concerning these examples are summarized by
Propositions \ref{prop42} and \ref{prop44}.
 Section \ref{Sec51} contains a generalization of results of the previous examples to the simplest type (ii) cases
 associated with $\SU(n)$ for any $n\geq 4$. In particular, we give the vertices of $\cA_y$ in Theorem \ref{thm51}
 for $y$ in the range \eqref{51}, and  describe the fibers over the $n(n-3)$ singular vertices in Theorem \ref{thm54}.
 Section \ref{Sec52}  contains miscellaneous further results about $\cA_y$ for other ranges of the parameter $y$ for various specific values of $n$.
 We conclude in Section \ref{Sec6} by giving a summary and an outlook towards open problems.
 Four appendices complement the main text.
 Appendix \ref{AppA} describes the connection to the compactified trigonometric Ruijsenaars--Schneider model.
  Appendix \ref{newAppB} elaborates some details regarding Proposition \ref{thmI1}.
 Appendix \ref{AppB} explains the irrelevant nature of various choices that appear in the construction in Section \ref{Sec31},
 and Appendix \ref{AppC} contains the outcome of numerical
 explorations of the polytope $\cA_y$. 
\end{conte*}

\section{Background material}\label{Sec2}

The construction of \cite{FK} was formulated in the setting of quasi-Hamiltonian geometry \cite{AMM}.
We below highlight the rudiments  of the general theory that are needed for our purpose,
focusing on the relevant example of the `internally fused double' of $\SU(n)$.
This prepares the ground  for presenting the description of the type (ii) momentum polytopes found in \cite{FK},
which will serve as the starting point of our investigation in this paper.

\subsection{Dynamics on the internally fused double}\label{Sec21}

A quasi-Hamiltonian structure for a compact Lie group $K$ is given by a $K$-manifold $P$
equipped with  a  $K$-invariant 2-form $\omega$
and an  equivariant momentum map $\mu: P \to K$, subject to a delicate set of axioms \cite{AMM}.
The equivariance refers to the conjugation action of $K$ on itself.
Here we do not detail the precise definition, just recall basic consequences \cite{AMM}
that make this concept useful from the perspective of Hamiltonian mechanics.

First, let $\mu_0$ be a regular value of the (quasi-Hamiltonian) momentum map $\mu$ and suppose that there
is only a single orbit type for the action of
 $K_{\mu_0}$ on $\mu^{-1}(\mu_0)$.
 Then $P(\mu_0) = \mu^{-1}(\mu_0)/K_{\mu_0}$ is a symplectic manifold with the symplectic form $\omega_{\mu_0}$
 verifying the equality
 \be
 \iota_0^* \omega = \pi_0^* \omega_{\mu_0},
 \label{T1}\ee
 where $\iota_0$ and $\pi_0$ are the canonical inclusion and projection maps defined on $\mu^{-1}(\mu_0)\subset P$.
 This coincides with the familiar relation of Marsden--Weinstein-Meyer symplectic reduction \cite{OR}, although
 $\omega$ is neither closed nor non-degenerate in general (but its kernel and exterior derivative
 are controlled by the axioms \cite{AMM}).
Second, one can associate a unique vector field $V_H$ on $P$
to any $K$-invariant function $H\in C^\infty(P)^K$ by requiring the usual relation
$\iota_{V_H} \omega = dH$ complemented by the condition that $D\mu(V_H)=0$ holds everywhere.
 The resulting `quasi-Hamiltonian vector field' $V_H$ is $K$-invariant and $\mu$ is constant along its integral curves.
Moreover, the projection of $V_H$ on $P(\mu_0)$  is the Hamiltonian
vector field of the `reduced Hamiltonian' $\cH \in C^\infty(P(\mu_0))$ that enjoys the relation
\be
H \circ \iota_0 =\cH \circ \pi_0.
\label{I.4}\ee
Third,  $C^\infty(P)^K$ is a genuine Poisson algebra with the Poisson bracket of two invariant functions
defined by $\{F,H\}:= dF(V_H)$.
In short, Hamiltonian mechanics and reduction works for invariant functions on quasi-Hamiltonian manifolds
 in the same
way as for all functions on symplectic manifolds.

Turning to the relevant example,
let us introduce the following inner product on the Lie algebra $\su(n)$,
\be
\langle X, Y\rangle := - \frac{1}{2} \tr(XY),\quad
\forall X,Y\in \su(n),
\label{T2}\ee
and define the 2-form $\omega$ on $P$ \eqref{I.1} by the formula
\be
\omega :=  \langle A^{-1} \mathrm{d}A \stackrel{\wedge}{,} \mathrm{d}B B^{-1}\rangle
+\langle  \mathrm{d}A A^{-1} \stackrel{\wedge}{,} B^{-1} \mathrm{d}B \rangle
- \langle (AB)^{-1} \mathrm{d} (A B) \stackrel{\wedge}{,} (BA)^{-1} \mathrm{d} (BA)
\rangle.
\label{T3}\ee
With this 2-form and the momentum map $\mu$ \eqref{I.2}, $P= \SU(n) \times \SU(n)$ becomes
a quasi-Hamiltonian manifold, known as the internally fused double of
$\SU(n)$.
It represents an analogue of the cotangent bundle
$T^*\SU(n)$ in the quasi-Hamiltonian world \cite{AMM}.

For any real function $h\in C^\infty(\SU(n))$,  define its $\su(n)$-valued derivative $\nabla h$  by
the relation
\be
\langle \nabla h(g), X \rangle = \dt h(e^{tX} g), \qquad \forall X \in \su(n).
\label{T4}\ee
Then, take an invariant function $h \in C^\infty(\SU(n))^{\SU(n)}$ and consider the functions $h_a, h_b \in C^\infty(P)^{\SU(n)}$ for which
\be
h_a(A,B):= h(A), \qquad
h_b(A,B) = h(B).
\label{T5}\ee
The integral curve of the quasi-Hamiltonian vector field of the function $h_a$ with initial value $(A_0, B_0)$ can be written as
\be
(A(t), B(t)) = (A_0, B_0 e^{-t \nabla h(A_0)})
\label{T6}\ee
and for $h_b$  the integral curve can be written as
\be
(A(t), B(t)) = (A_0 e^{t \nabla h(B_0)}, B_0).
\label{T7}\ee
Letting $\fH_a$ and $\fH_b$ denote the sets of functions $\{h_a\}$ and $\{h_b\}$, we see that these
are two Abelian Poisson algebras of $C^\infty(P)^{\SU(n)}$, and both have functional dimension $(n-1)$.
Upon reduction, they give rise to the Abelian Poisson algebras  $\hat \fH_a$ and $\hat \fH_b$
on any quotient space $P(\mu_0)=\mu^{-1}(\mu_0)/\SU(n)_{\mu_0}$, which is in general a stratified symplectic space.
It is clear that these two Abelian Poisson algebras are equivalent.
In fact, they can be converted into each other by a certain `duality automorphism' \cite{AMM,FKlim} of the double $P$ \eqref{I.1}.

\subsection{Action variables before reduction}\label{Sec22}

Below, we recall a  parametrization of the conjugacy classes of $\SU(n)$ and utilize it for defining
$(n-1)$ invariant functions on $P$ \eqref{I.1}
that are smooth on a dense open subset and
generate  mutually commuting  \emph{$2\pi$-periodic flows}.

It will be convenient for us to formulate the parametrization using the  $(n-1)$-dimensional  simplex $\cA \subset \bR^n$ given by
\be
\cA = \{ \xi = (\xi_1, \dots, \xi_n) \in \bR^n  \mid  \xi_1+\cdots + \xi_n = \pi,\,\, \xi_i\geq 0,\,\forall i=1,\dots, n\}.
\label{T8}\ee
It is well known that  any $g\in \SU(n)$ is conjugate to a unique diagonal matrix of the form
\be
\delta(\xi):= \exp\left(-2 \ri \sum_{j=1}^{n-1} \xi_j \lambda_j\right),  \quad \xi\in \cA,
\label{T9}\ee
where the $\lambda_j$ ($j=1,\dots, n-1$) are the $n\times n$ matrices
\be
\lambda_j:= \sum_{k=1}^{j} E_{k,k} - \frac{j}{n} \1_n
\label{T10}\ee
with the elementary diagonal matrices $E_{k,k}$ and the unit matrix $\1_n$.
The matrix $\delta(\xi)$ can also be  displayed as
\be
\delta(\xi) = \diag(\delta_1(\xi), \ldots, \delta_n(\xi)),
\quad
\delta_1(\xi):=e^{\frac{2\ri}{n}\sum_{j=1}^n j\xi_j}
\quad
\delta_{k+1}(\xi):= e^{2\ri \xi_k} \delta_k(\xi),
\label{T11}\ee
where $\delta_{n+1} := \delta_1$.
This is a special case of
the parametrization of the conjugacy classes of any connected and simply connected, compact  Lie group with simple Lie algebra by a Weyl alcove \cite{DK}.
(In the general case, the matrices $\lambda_j$ are replaced by the fundamental coweights of the corresponding Lie algebra.)

Now, writing $g_1 \sim g_2$ if the two matrices $g_1$ and $g_2$ are similar,
we define the surjective $\SU(n)$ invariant map $\Xi: \SU(n) \to \cA$ by setting
\be
\Xi(g) := \xi \quad \hbox{if} \quad g \sim \delta(\xi).
\label{T12}\ee
The dense subset $\SU(n)_\reg$ of regular elements, matrices with $n$ distinct eigenvalues, is mapped onto the
interior of $\cA$, denoted  $\cA^\reg$.
The map $\Xi$ is globally continuous and its restriction on
$\SU(n)_\reg$ is $C^\infty$ (even real-analytic).  At any point of $\Xi^{-1}(\partial \cA)$, some component functions
of $\Xi$ lose their differentiability.
In fact \cite{DK}, the map $\Xi$ descends to a homeomorphism from $\SU(n)/\Ad_{\SU(n)}$ onto $\cA$
and its restriction to regular elements  yields a diffeomorphism from $\SU(n)_\reg/ \Ad_{\SU(n)}$ onto $\cA^\reg$.

Next, we introduce the map $\beta: P \to \cA$ by
\be
\beta(A,B):= \Xi(B).
\label{T13}\ee
The independent component functions are $\beta_1,\dots, \beta_{n-1}$;  $\beta_n$ is included  for notational convenience.
These functions are invariant under the conjugation action of $\SU(n)$, and their quasi-Hamiltonian vector fields
generate mutually commuting flows on $\SU(n) \times \SU(n)_\reg$.
To describe the flows, for $t = (t_1,\dots, t_{n-1}) \in \bR^{n-1}$ we let $\varrho(t)$ denote the following element
of the standard maximal torus of $\SU(n)$:
\be
\varrho(t):= \exp \left(\ri \sum_{j=1}^{n-1}t_j ( E_{j+1,j+1} - E_{j, j})\right),
\label{T14}\ee
and for every $B \in \SU(n)_\reg$ choose a `diagonalizer' $g(B) \in \U(n)$  for which
\be
g(B) B g(B)^{-1} = \delta(\Xi(B)).
\label{T15}\ee
Notice that $\varrho$ gives an isomorphism from $\bT^{n-1} = \bR^{n-1}/ (2\pi \bZ)^{n-1}$ onto
the maximal torus of $\SU(n)$.
Then,  we find that the joint flow of the commuting  quasi-Hamiltonian vector fields $V_{\beta_j}$
($j=1,\dots, n-1$)
 operates on the initial value $(A,B)$  according to the map
\be
\Psi^b_t:  (A, B) \mapsto ( A g(B)^{-1} \varrho(t) g(B), B).
\label{T16}\ee
This represents a free action of
$\bT^{n-1}$ on $\SU(n) \times \SU(n)_\reg$.

Since $\beta_j = \Xi_j \circ \pi_2$, the formula \eqref{T16} follows from \eqref{T7} and the equality
\be
\nabla \Xi_j(B) = g(B)^{-1} (\ri (E_{j+1, j+1} - E_{j,j})) g(B),
\qquad
\forall j=1,\dots, n-1,\,\,  \forall B\in \SU(n)_\reg,
\label{T17}\ee
which is easily verified.
Although $g(B)$ is well-defined only up to left-multiplication by an arbitrary diagonal unitary matrix, the ambiguity drops out from the \eqref{T17}.
If desired, one may take $g(B)$ from $\SU(n)$, and locally it can be chosen as a smooth function of $B\in \SU(n)_\reg$.

We can view the functions  $\beta_1,\dots, \beta_{n-1}$ as  \emph{action variables} associated
with the Abelian Poisson algebra $\fH_b = \pi_2^* C^\infty(\SU(n))^{\SU(n)}$ consisting
of globally smooth invariant Hamiltonians.
Indeed, it is immediate from the definitions that any $H \in \fH_b$ can be expressed
in terms of these functions,  as $H = \chi \circ \beta$ with a uniquely defined function $\chi$ on $\cA$,
and the  $\beta_j$ generate mutually commuting $2\pi$-periodic flows on $\beta^{-1}(\cA^\reg)$.
The only non-standard feature is that $P$ is not a symplectic, but a quasi-Hamiltonian manifold.

Of course, one can similarly define the map $\alpha(A,B):= \Xi(A)$, which can be used to generate
a free $\bT^{n-1}$ action on $\SU(n)_\reg \times \SU(n)$ and serves as a collection of action
variables behind the Abelian Poisson algebra $\fH_a =\pi_1^* C^\infty(\SU(n))^{\SU(n)}$.

 \subsection{The type (ii) momentum polytopes}\label{Sec23}

 Let us suppose that $0<y < \pi$ is an admissible value, so that
 \be
 P(\mu_0(y)) = \mu^{-1}(\mu_0(y))/ \SU(n)_{\mu_0(y)}
 \label{T18}\ee
 is a connected compact symplectic manifold of dimension $2(n-1)$.
 In all these cases, the Abelian Poisson algebras $\fH_a$ and $\fH_b$ reduce to Abelian
 Poisson algebras of functional dimension $(n-1)$ on $P(\mu_0(y))$, denoted $\hat \fH_a$ and $\hat \fH_b$,
 which represent two  Liouville integrable systems.
 These two systems are equivalent and we focus on $\hat \fH_b$.

 Let us define the map $\hat \beta: P(\mu_0(y)) \to \cA$ by the prescription
 \be
 \beta \circ \iota_0 = \hat \beta \circ \pi_0,
 \label{T19}\ee
 where the notation is the same as in \eqref{I.4}.
 Then, the map
 \be
 (\hat \beta_1, \dots, \hat \beta_{n-1}): P(\mu_0(y)) \to \bR^{n-1}
 \label{T20}\ee
serves as an \emph{action map} for the integrable system
 having the globally smooth commuting Hamiltonians that form $\hat \fH_b$.
This means that all elements of $\hat \fH_b$ can be expressed in terms of the
functions $\hat \beta_1,\dots, \hat \beta_{n-1}$,  and they generate a Hamiltonian action of $\bT^{n-1}$ on
the dense\footnote{It is not entirely obvious that  $\hat \beta^{-1}(\cA^\reg) \subset P(\mu_0(y))$ is
always dense, but this is proved in \cite{FK}.}  open subset $\hat \beta^{-1}(\cA^\reg)$ of $P(\mu_0(y))$.
Instead of the map \eqref{T20}, we also call the equivalent map $\hat \beta$ \eqref{T19} \emph{action map}.
 Moreover, the terms `action map' and `momentum map' are used interchangeably throughout the paper.
These issues should not cause any confusion.

 Recall that an admissible $y$ is called type (ii) value if $\hat \beta^{-1}(\cA^\reg) \subset P(\mu_0(y))$  is a
 \emph{proper subset}.
 The complementary type (i) cases are characterized by Theorem \ref{thmI1}, and are completely understood.
 Considering $\SU(n)$ and a type (ii) parameter $y$, we have to assume that $n\geq 4$ and
 \be
 k \frac{\pi}{n} <  y < (k+1) \frac{\pi}{n}  \quad \hbox{for some} \quad k \in \{ 1, 2,\dots, n-2\}.
 \label{T21}\ee
 The precise range of type (ii) parameters, including for which intervals \eqref{T21}  they occur, is
 given (somewhat implicitly) by Theorem \ref{thmI1}.
Next, we recall  a characterization of the image
\be
\cA_y:= \hat \beta(P(\mu_0(y))).
\label{T22}\ee

\begin{thm}[\cite{FK}]\label{thmT1}
Consider $P(\mu_0(y))$ for $n\geq 4$ with a type (ii) admissible value of $y$ satisfying \eqref{T21}.
Introduce the $(n-1)$-dimensional hyperplane $E \subset \bR^n$ as
\be
E:= \{ \xi \in \bR^n \mid \xi_1 + \dots + \xi_n = \pi \}
\label{T23}\ee
and  extend the definition of the coordinates $\xi_j$ $n$-periodically, i.e., set $\xi_j = \xi_{j+n}$ for all $j\in \bZ$.
Then, the image $\cA_y$ of the action map $\hat \beta$ is the set of points of $E$ for which the
following $2n$ inequalities hold:
\be
\xi_\ell + \xi_{\ell +1} + \cdots + \xi_{\ell + k - 1} \leq y
\quad \hbox{and} \quad
\xi_\ell + \xi_{\ell +1} + \cdots + \xi_{\ell + k } \geq y, \quad \forall \ell=1,\dots,n.
\label{T24}\ee
\end{thm}

\begin{rem}\label{remT2}
The two inequalities for fixed $\ell$ imply hat $\xi_{\ell + k} \geq 0$, i.e., the convex polyhedron in $E$ defined by the inequalities
\eqref{T24}
is contained in $\cA$ \eqref{T8}, as it must be the case.
 In particular, $\cA_y$ is a polytope, since it is bounded.
The interior of $\cA_y$ contains the point $\xi^* \in \cA_\reg$ whose components are
\be
\xi^*_j = \frac{\pi}{n} \quad\hbox{ for all}\quad  j=1,\dots,n.
\label{xistar}\ee
Indeed, this point satisfies \eqref{T24} with strict inequalities, and thus $\cA_y$ also contains a neighbourhood of $\xi^*$ in $E$ \eqref{T23}.
By the definition of $y$ being type (ii), the $(n-1)$-dimensional polytope $\cA_y$ intersects the boundary $\partial \cA$.
 Define the action of the cyclic permutation $\sigma \in {\mathrm{S}}_n$ on $E$ by
\be
\sigma(\xi)_j := \xi_{j+1},
\label{T25}\ee
and observe that $\sigma$ preserves the inequalities \eqref{T24}. As a result, $\sigma$ maps $\cA_y$ to $\cA_y$,
and it maps vertices to vertices, edges to edges and so on. This cyclic permutation symmetry of the polytope $\cA_y$  goes
back to the identity
\be
\Xi (g e^{-2\ri \pi/n} ) = \sigma(\Xi(g)), \quad \forall g\in \SU(n),
\label{T26+}\ee
enjoyed by the map $\Xi$ \eqref{T12}, and the  fact that if $(A,B) \in \mu^{-1}(\mu_0(y))$ then so is $(A, B e^{-2\ri \pi/n})$.
\end{rem}

For completeness, let us recall from \cite{FK} what happens
in the type (i) cases, when $y$ varies between $k \frac{\pi}{n}$ with $\gcd(k, n)=1$
and one of its neighbours in  the Farey sequence $\fF_n$.
In fact, then the  conditions \eqref{T24} are still valid, but those $n$ inequalities
that contain sums of $k$ terms imply
the other $n$ inequalities, and  $\cA_y$ becomes a simplex in the interior of $\cA$ \eqref{T8}.

\section{The structure of the fibers of the momentum map }\label{Sec3}

Let us decompose the polytope $\cA_y = \hat \beta(P(\mu_0(y)))$ as
\be
\cA_y = \cA_y^o \cup \partial \cA_y.
\label{G30}\ee
The interior $\cA_y^o$ lies in $\cA^\reg$, while we  further decompose
the boundary $\partial \cA_y$ as
\be
\partial \cA_y = (\partial \cA_y)^\reg \cup (\partial \cA_y)^{\sing}
\quad \hbox{with}\quad (\partial \cA_y)^\sing := \partial \cA_y \cap \partial \cA.
\label{G31}\ee
For $\xi \in \cA_y^o$, the fiber $\hat \beta^{-1}(\xi)$ is a standard Liouville torus $\bT^{n-1}$, and
for $\xi \in (\partial \cA_y)^\reg$ it is an isotropic  torus of smaller dimension.
This follows directly from the results of \cite{FK},
noting that the $\bT^{n-1}$ action \eqref{T16} descends to $\hat \beta^{-1}(\cA_y \cap \cA^\reg)$.
Being orbits of a torus action, these fibers are embedded submanifolds.

\subsection{The singular fibers are smooth isotropic submanifolds}\label{Sec31}

In this section, we will prove that also \emph{ the `singular fibers'
\be
\hat \beta^{-1}(\xi) = \left(\beta^{-1}(\xi) \cap \mu^{-1}(\mu_0(y))\right)/\SU(n)_{\mu_0(y)}
\quad \hbox{for}\quad
\xi \in (\partial \cA_y)^\sing
\label{B1}\ee
are embedded isotropic submanifolds of the phase space $P(\mu_0(y))$. }
More specifically, we shall see that
every fiber $\hat \beta^{-1}(\xi)$ can be identified as a quotient space of
the isotropy group $\SU(n)_{\delta(\xi)}$ by an action of  a subgroup.
Our construction will involve certain choices, but those will turn out irrelevant,
since a subset can only have a single embedded submanifold structure.

Below,  it is assumed that $0<y<\pi$ is a type (ii) admissible parameter.
We shall first consider the corresponding fiber $\beta^{-1}(\xi)$ in $\mu^{-1}(\mu_0(y))$.
The final outcome of our analysis is given by Theorems \ref{thm314} and \ref{thm316} together with Lemma \ref{lem310}, using the
notations introduced in Definition \ref{defn36}.
 The reader will notice that the subsequent analysis is valid  for any  $\xi \in \cA_y$.
It is a good check that for  $\xi \in  (\cA_y \cap \cA^\reg)$  we shall recover the known result.

For any column vector $u\in \bC^n$ of unit length, we introduce
\be
\hat \mu(y,u):= e^{2\ri y}\1_n + (e^{-2(n-1) \ri y} - e^{2\ri y}) u u^\dagger,
\label{G1}\ee
which is an element of $\SU(n)$.
Incidentally, for $e^{2n \ri y}\neq 1$,  $\hat \mu(y,u)$ as given by \eqref{G1}  has unit determinant precisely if $u^\dagger u =1$.

\begin{lem}\label{lemG1}
For any $(A,B) \in \mu^{-1}(\mu_0(y)) \cap \beta^{-1}(\xi)$ there exist elements $u\in \bC^n$ and $g\in \U(n)$ with $g_{jn}=u_j$ for $j=1,\dots, n$ such that
$\delta(\xi) \sim \hat \mu(y,u) \delta(\xi)$ and the pair $(A,B)$ has the form
\be
(A,B) =( g^{-1} \hat A g , g^{-1} \delta(\xi) g),
\label{G2}\ee
where $\hat A$ is an element of $\SU(n)$ satisfying the condition
\be
\hat A \delta(\xi) \hat A^{-1} = \hat \mu(y,u) \delta(\xi).
\label{G3}\ee
\end{lem}
\begin{proof}
Suppose that
\be
AB A^{-1} B^{-1} = \mu_0(y)
\label{G4}\ee
and
 take a unitary matrix for which $g B g^{-1} = \delta(\xi)$. Defining $\hat A := g A g^{-1}$ and letting $u$ denote the last column of $g$,
 the condition \eqref{G4} is readily seen to be equivalent to \eqref{G3}, because one has due to \eqref{I.5}:
 \be
 g \mu_0(y) g^{-1}  = \hat \mu(y,u).
 \label{G5}\ee
 \end{proof}

\begin{lem}\label{lemG2}
The intersection  $\mu^{-1}(\mu_0(y))\cap \beta^{-1}(\xi)$ is the set of elements of the form
\be
(A,B) = (g^{-1} \hat A_u T g, g^{-1} \delta(\xi) g),
\label{G6}\ee
where $g\in \U(n)$ is an arbitrary unitary matrix whose last column $u$ satisfies the similarity relation
\be
\delta(\xi) \sim \hat \mu(y,u) \delta(\xi),
\label{G7}\ee
$\hat A_u \in \SU(n)$ is any fixed solution of \eqref{G3},
and $T\in \SU(n)_{\delta(\xi)}$ is arbitrary.
\end{lem}
\begin{proof}
By Lemma \ref{lemG1}, any $(A,B)\in \mu^{-1}(\mu_0(y))\cap \beta^{-1}(\xi)$ has the form \eqref{G6} since the solutions
of \eqref{G3} are the elements $\hat A = \hat A_u T$ with a particular solution $\hat A_u$ and any $T\in \SU(n)_{\delta(\xi)}$.

To see the converse, consider $(A,B)$ given \eqref{G6} and notice that
\be
 (\hat A_u T) \delta(\xi) (\hat A_u T)^{-1} = \hat \mu(y,u) \delta(\xi)
 \label{G8}\ee
 holds. Conjugating this by $g^{-1}$, and defining $(A,B)$ by \eqref{G6},  gives
 \be
 A B A^{-1} = (g^{-1} \hat \mu(y,u) g) B,
 \label{G9}\ee
 which is just the condition \eqref{G4} since $\hat \mu(y,u) = g \mu_0(y) g^{-1}$ for any $g\in \U(n)$ having $u$ as
 its last column. Therefore, any $(A,B)$ given by \eqref{G6} belongs to $\mu^{-1}(\mu_0(y))\cap \beta^{-1}(\xi)$.
 \end{proof}

It is worth stressing that  in Lemma \ref{lemG2}  $u$ is only constrained by \eqref{G7}.
For fixed $u$, $\hat A_u$ is chosen, while $g\in \U(n)$ still varies subject to $g_{jn}=u_j$, $\forall j$.
To control the fiber, we have to pin down  the possible vectors $u$  and the possible matrices $g$.

\begin{lem}\label{lemG3}
For fixed $\xi \in \cA_y$,  let us choose a vector
$u_0\in \bC^n$ that verifies the relation
\be
\delta(\xi) \sim \hat \mu(y,u_0) \delta(\xi).
\label{G10}\ee
Then,  $u\in \bC^n$ satisfies \eqref{G7} if and only if it can be written in the form
\be
u = X u_0
\quad \hbox{with}\quad X \in \U(n)_{\delta(\xi)}.
\label{G11}\ee
\end{lem}
\begin{proof}
The relations \eqref{G10} and \eqref{G11}  imply immediately that $u=X u_0$ satisfies \eqref{G7}.
 We have to prove the converse, i.e., that any $u$ subject to \eqref{G7} has the form
in \eqref{G11}. For this purpose, note that the equality of the characteristic polynomials of $\delta(\xi)$
and $\hat \mu(y,u) \delta(\xi)$ is equivalent to the similarity relation \eqref{G7}, and it
 can be  spelled out as
\be
\prod_{j=1}^n (\delta_j(\xi) -\zeta )=\prod_{j=1}^n (\delta_j(\xi) e^{2\ri y}- \zeta)
+ (e^{2\ri(1-n) y}-e^{2\ri y}) \sum_{k=1}^n \Bigl(\vert u_k \vert^2 \delta_{k}(\xi)
\prod_{\substack{j=1\\j \neq k}}^n (\delta_j (\xi)e^{2\ri  y} -\zeta)\Bigr).
\label{E1}\ee
This results by calculating $\det(\hat \mu(y,u) - \zeta \1_n)$
 using the well known identity
 $\det(\1_n + v w^T )= 1 + w^T v$, which holds for any column vectors $v, w \in \bC^n$.
Alternatively, \eqref{E1} equivalent to the following equality of rational functions on the complex plane:
\be
\prod_{j=1}^n \frac{\delta_j(\xi) -\zeta }{\delta_j(\xi)e^{2\ri  y} -\zeta}=
1 +  (e^{2\ri(1-n) y}-e^{2\ri y}) \sum_{k=1}^n \frac{\vert u_k \vert^2 \delta_{k}(\xi)}{\delta_k (\xi)e^{2\ri  y} -\zeta}.
\label{E2}\ee
In turn, this can be rewritten as
\be
\frac{\prod_{j=1}^n (\delta_j(\xi) -\zeta) - \prod_{j=1}^n (\delta_j(\xi)e^{2\ri  y} -\zeta)}{ \prod_{j=1}^n (\delta_j(\xi)e^{2\ri  y} -\zeta)}
= (e^{2\ri(1-n) y}-e^{2\ri y}) \sum_{k=1}^n \frac{\vert u_k \vert^2 \delta_{k}(\xi)}{\delta_k (\xi)e^{2\ri  y} -\zeta}.
\label{E3}\ee
Consequently, $\xi$ must be such that the left-hand side has only first order poles,   at the distinct values of $\delta_j(x)^{2\ri y}$,
and $u$ is constrained by the equality of the residues at those poles.

If $\xi \in \cA_y \cap \cA^\reg$, then $n$ distinct poles occur, and by comparing the residues one obtains
\be
\vert u_\ell \vert^2 = z_\ell(\xi,y), \qquad \forall \ell=1,\dots, n,
\label{E4}\ee
with the functions
\be
z_\ell(\xi,y) :=
      \frac{\sin(y)}{\sin(ny)}
      \prod^n_{\substack{j=1\\j \neq \ell}} \frac{e^{\ri y}\delta_\ell(\xi) - e^{-\ri y}\delta_j(\xi)}
                           {\delta_\ell(\xi) - \delta_j(\xi)}= \frac{\sin(y)}{\sin(ny)}
\prod_{j=\ell+1}^{\ell+ n-1}
\left[  \frac{\sin(\sum_{m=\ell}^{j-1}\xi_m - y)}{\sin(\sum_{m=\ell}^{j-1} \xi_m)} \right].
\label{E5}\ee
Observe that these functions are strictly positive in the interior of the polytope $\cA_y$, where all the inequalities \eqref{T24}
hold in the strict sense. (In fact \cite{FK}, the property $z_\ell(\xi,y)>0$, $\forall \ell$ characterizes the interior of the polytope $\cA_y$.)

Now, suppose that $\xi \in \cA_y \cap \partial \cA$, i.e., $\delta(\xi)$ is not a regular element.
Then, $\delta(\xi)$ has $1\leq r<n$ distinct eigenvalues, which we denote
as $\Delta_1(\xi),\dots, \Delta_r(\xi)$.
Furthermore, suppose that $\Delta_k(\xi)$ occurs with multiplicity $m_k$, $m_1+\cdots + m_r=n$,
 and let $\cM_k \subset \{1,\dots, n\}$ be the subset of $m_k$ indices satisfying
\be
\delta_j(\xi) = \Delta_k(\xi), \quad \forall j\in \cM_k.
\label{E6}\ee
In this case, the equality of the residues gives
 \be
(e^{2\ri(1-n) y}-e^{2\ri y})\Delta_k(\xi)  \sum_{j \in \cM_k}\vert u_j \vert^2 =
-\mathrm{Res}_{\zeta=\zeta_k} F(\zeta, \xi, y), \quad \hbox{for}\quad \zeta_k :=  e^{2\ri y} \Delta_k(\xi),
\label{E7}\ee
where $F(\zeta,\xi,y)$ is the function that appears on the left-hand side of equation \eqref{E3}.

Notice that $\U(n)_{\delta(\xi)} = \bT^n$, the diagonal subgroup, if $\delta(\xi)$ is regular.
If $\delta(\xi)$ is not regular, then (possibly up to an internal automorphism of $\U(n)$ that
brings all equal eigenvalues of $\delta(\xi)$  into neighbouring positions\footnote{This is needed if
$\xi_1+ \cdots + \xi_j =\pi$ for some $j\leq (n-1)$, since in such cases $\delta_1(\xi) = \delta_n(\xi)$.}) one has
\be
\U(n)_{\delta(\xi)} \simeq \U(m_1) \times \U(m_2) \times \cdots \times \U(m_r).
\label{E8}\ee
In either case, the form of the constraints \eqref{E4}  and \eqref{E7} implies immediately that any two vectors
$u\in \bC^n$ that solve the similarity relation \eqref{G7}
can be transformed into each other by multiplication by an element of $\U(n)_{\delta(\xi)}$.
Indeed, this follows since any two $\bC^m$ vectors of the same length can be transformed into each other
by an element of $\U(m)$.
\end{proof}

It is worth noting that the equality $\sum_{\ell =1}^n  \vert u_\ell \vert^2 =1$ follows by evaluation of \eqref{E1}
at $\zeta=0$ for any $\xi \in \cA_y$, and the sums $\sum_{j \in \cM_k}\vert u_j \vert^2$ can be computed
 from \eqref{E3}, \eqref{E7} for any $\xi \in \cA_y$.
We shall also use that, in consequence of \eqref{E4},  $\sum_{\ell=1}^n z_\ell(\xi,y) =1$ holds for
 $\xi\in \cA_y \cap \cA^\reg$.

\begin{rem}
The formula \eqref{E2} is a particular case of a so-called Weinstein-Aronszajn determinant, which also works for certain operators in Banach spaces, see \cite{Kato}.
If $T = T_0 + V$ is a finite rank perturbation of $T_0$ then the resolvent $R(\zeta,T)$ of $T$ can
be written in terms of the resolvent $R(\zeta, T_0)$ of the unperturbed operator $T_0$ as
\be
   R(\zeta,T) = R(\zeta,T_0) - R(\zeta, T_0) V ( I + VR(\zeta,T_0))^{-1} R(\zeta, T_0)
\ee
and $\omega(\zeta) = \det( I + V R(\zeta, T_0))$ is the Weinstein-Aronszajn determinant.
It is a meromorphic function of $\zeta$ and its residues determine the spectral projections. The Laurent expansion of
 $\omega(\zeta)$ has poles of higher order when the unperturbed operator has eigenvalues with multiplicity.
 The simplest case of a rank one perturbation leads to the present formula \eqref{E2}. The Weinstein-Aronszajn determinant has been used to great effect in
 the work of Moser \cite{Moser79ChernSymposium} on integrable systems.
\end{rem}

\begin{lem}\label{lemG4}
Keeping $u_0$ from Lemma \ref{lemG3}, let $g_0$ be a unitary matrix having $u_0$ as its last column.
Then, the unitary matrices with last column $u = X u_0$ \eqref{G11} are given by the expression
\be
g = X g_0 \eta^{-1},
\label{G12}\ee
where $\eta \in \U(n)$ verifies $\eta_{j,n} = \eta_{n,j} = \delta_{j,n}$ for all $j=1,\dots, n$.
\end{lem}
\begin{proof}
Any unitary matrix $g$ can be written as in \eqref{G12} with
$\eta = g^\dagger Xg_0$. If $g$ and $X g_0$ have the same last column, then we have
\be
\eta_{j,n} = \sum_{k=1}^n \bar g_{k,j} (Xg_0)_{k,n} = \sum_{k=1}^n \bar g_{k,j} g_{k,n} = \delta_{j,n},
\label{G15}\ee
because $g$ is unitary.  Then, the relation $\eta_{n,j} = \delta_{n,j}$ follows from the unitarity of $\eta$.
It is also obvious that $g$ and $Xg_0$ have the same last column if $\eta_{j,n} = \delta_{j,n}$ is valid.
\end{proof}

\begin{rem}\label{remG5}
Note that $\eta$ in \eqref{G12} satisfies $\eta \mu_0(y) \eta^{-1} = \mu_0(y)$ and conjugation by $\eta$
is the same as conjugation by $\tilde \eta \in \SU(n)_{\mu_0(y)}$ given by
$\tilde \eta =  c \eta$ with any $c\in \bC$ verifying $c^n = (\det \eta)^{-1}$.
\end{rem}

\begin{defn}\label{defn36}
For any fixed $\xi \in \cA_y$, choose a unitary matrix $g_0$ whose last column $u_0$ satisfies the similarity relation
$\delta(\xi) \sim \hat \mu(y,u_0) \delta(\xi)$ and also choose an element $\hat A_0\equiv \hat A_{u_0}\in \SU(n)$ verifying
 $\hat A_0 \delta(\xi) \hat A_0^{-1}= \hat \mu(y,u_0) \delta(\xi)$.
 Then, define  the subset $M(\xi)  \subset \beta^{-1}(\xi)\cap \mu^{-1}(\mu_0(y))$ by
 \be
M(\xi):= \{ (g_0^{-1} \hat A_0 T g_0, g_0^{-1} \delta(\xi) g_0) \mid T \in \SU(n)_{\delta(\xi)} \},
\label{G20}\ee
and  introduce the closed Lie subgroups
\be
\SU(n)^{u_0}_{\delta(\xi)}:= \{ \eta \in \SU(n)\mid \eta \delta(\xi) \eta^{-1} = \delta(\xi),\, \eta \hat\mu(y, u_0) \eta^{-1} = \hat \mu(y, u_0)\}
< \SU(n)
\label{SUnu0}\ee
and
\be
K^{g_0}_\xi := g_0^{-1} \SU(n)^{u_0}_{\delta(\xi)} g_0 < \SU(n)_ {\mu_0(y)}.
\label{Ku0}\ee
\end{defn}

\begin{rem}
 We suppressed the dependence of $M(\xi)$ on the choices of $g_0$ and $\hat A_0$, which will turn out to be irrelevant,
 as is further discussed in Appendix \ref{AppB}.
The notations introduced in \eqref{SUnu0} and \eqref{Ku0} are just convenient shorthands introduced
in order to economize on writing out the intersections of isotropy groups:
\be
\SU(n)^{u_0}_{\delta(\xi)} \equiv \SU(n)_{\delta(\xi)} \cap \SU(n)_{\hat \mu(y,u_0)}
\quad \hbox{and}\quad
K_\xi^{g_0} \equiv \SU(n)_{g_0^{-1} \delta(\xi) g_0} \cap \SU(n)_{\mu_0(y)}.
\ee
Later we shall use that $\SU(n)^{u_0}_{\delta(\xi)}$ consists of the elements of
$\SU(n)_{\delta(\xi)}$ that have $u_0$ as an eigenvector. This holds since
$u_0 u_0^\dagger = \eta u_0 (\eta u_0)^\dagger$ for $\eta \in \U(n)$ is equivalent to $\eta u_0 =  \Gamma u_0$ for some $\Gamma \in \U(1)$.
\end{rem}

\begin{lem}\label{lemG8}
The action of the subgroup $K^{g_0}_\xi < \SU(n)_{\mu_0(y)}$ \eqref{Ku0} preserves $M(\xi)$ \eqref{G20}.
Even more, if $(A,B)\in M(\xi)$ and $\eta \in \SU(n)_{\mu_0(y)}$  are such that
$(\eta A \eta^{-1}, \eta B \eta^{-1})$ belongs to $M(\xi)$, then $\eta\in  K^{g_0}_\xi$.
\end{lem}
\begin{proof}
According to Lemma \ref{lemG2}, $M(\xi)$ contains all elements $(A,B) \in \mu^{-1} (\mu_0(\xi)) \cap \beta^{-1}(\xi)$
for which $B = g_0^{-1} \delta(\xi) g_0$, which immediately implies the first statement.
Notice that if $(A,B)$ and $ (\eta A\eta^{-1} ,\eta  B \eta^{-1})$ are in  $M(\xi)$,
then $\eta \in g_0^{-1} \SU(n)_{\delta(\xi)} g_0$.
Since $K_\xi^{g_0} = g_0^{-1} \SU(n)_{\delta (\xi)} g_0 \cap  \SU(n)_{\mu_0(y)}$, the second statement follows.
\end{proof}

\begin{lem}
Let us make $M(\xi)$ \eqref{G20} into a smooth manifold via its parametrization by the elements of $\SU(n)_{\delta(\xi)}$.
Then, the mapping
\be
j_\xi: M(\xi) \to \mu^{-1}(\mu_0(y))
\label{B4}\ee
given by the tautological injection is a smooth embedding. The restriction of the action of $\SU(n)_{\mu_0(y)}$ on $M(\xi)$
yields  a smooth action
\be
K_\xi^{g_0} \times M(\xi) \to M(\xi),
\label{B7}\ee
and the quotient space
$M(\xi)/K_\xi^{g_0}$ admits a unique smooth manifold structure for which the canonical projection
$\pi_\xi: M(\xi) \to M(\xi)/K_\xi^{g_0}$ is a smooth submersion.
\end{lem}
\begin{proof}
Consider the chain of maps
\be
M(\xi) \overset{j_\xi}{\longrightarrow} \mu^{-1}(\mu_0(y)) \overset{\iota_0}{\longrightarrow} P,
\label{B5}\ee
and notice that the composed map $\iota_0 \circ j_\xi$ is an embedding.
This implies that the map $j_\xi$ is smooth.
One may show this directly, or by applying the universal property of initial submanifolds to $\mu^{-1}(\mu_0(y)) \subset P$,
 since every embedded submanifold is also an initial  submanifold \cite{OR}.

The smooth map $D j_\xi$ has full rank at every point of $M(\xi)$, since the map $D (\iota_0 \circ j_\xi)$ has full rank
at every point of $M(\xi)$,  because $\iota_0 \circ j_\xi$ is an embedding.
The map $j_\xi$ is obviously injective, and thus it is an injective immersion.
It is then an embedding, since every injective immersion of a compact manifold is an embedding \cite{L}.

The action map $K_\xi^{g_0} \times M(\xi) \to \mu^{-1}(\mu_0(y))$ is smooth, since it is the
restriction of the action map $\SU(n)_{\mu_0(y)} \times \mu^{-1}(\mu_0(y)) \to \mu^{-1}(\mu_0(y))$ on an embedded submanifold.

The claim about the smooth manifold structure of the quotient space $M(\xi)/K_\xi^{g_0}$ follows from standard results \cite{L,OR},
since the isotropy group of every point of $M(\xi)$ is the center $\bZ_n < \SU(n)$.
\end{proof}

We note in passing that for $\xi \in \cA_y \cap \cA^\reg$ the manifold $M(\xi)$ is an orbit of the free $\bT^{n-1}$ action
as given in \eqref{T16}.

\begin{lem}\label{lem310}
The parametrization of $M(\xi)$ \eqref{G20}  by $\SU(n)_{\delta(\xi)}$ leads to an identification
\be
M(\xi)/ K_\xi^{g_0} \equiv \SU(n)_{\delta(\xi)}/ \SU(n)_{\delta(\xi)}^{u_0}.
\label{id310}\ee
The quotient space on the right-hand side refers to the group action
\be
\SU(n)_{\delta(\xi)}^{u_0} \times \SU(n)_{\delta(\xi)} \to \SU(n)_{\delta(\xi)}
\label{act31}\ee
given by the map
\be
(X, T) \mapsto \Theta_0(X) T X^{-1}, \qquad \forall (X,T) \in  \SU(n)^{u_0}_{\delta(\xi)} \times \SU(n)_{\delta(\xi)},
\label{act33}\ee
where
\be
\Theta_0(X) := \hat A_0^{-1} X \hat A_0 \in \SU(n)_{\delta(\xi)}.
\label{G23}\ee
The manifold \eqref{id310} is compact and connected.
\end{lem}
\begin{proof}
This is a straightforward  consequence of the equality
\be
(g_0^{-1} X g_0) (g_0^{-1} \hat A_0 T g_0) (g_0^{-1} X g_0)^{-1} = g_0^{-1} (\hat A_0 \Theta_0(X) T X^{-1}) g_0,
\ee
which ensures that the group actions \eqref{B7}  and \eqref{act31} are intertwined by the equivariant diffeomorphism
\be
K_\xi^{g_0} \times M(\xi) \ni (g_0^{-1} X g_0, (g_0^{-1} \hat A_0 T g_0, g_0^{-1} \delta(\xi) g_0 ))
\mapsto (X, T) \in \SU(n)_{\delta(\xi)}^{u_0} \times \SU(n)_{\delta(\xi)}.
\ee
The fact that $\Theta_0(X)$ belongs to  $\SU(n)_{\delta(\xi)}$ is readily checked from the definitions.
 The quotient space \eqref{id310} is compact and connected since
$\SU(n)_{\delta(\xi)}$ has these properties.   Indeed, it is well known (\cite{Bourbaki}, Theorem 1 in Section 9.5.3)
that every isotropy subgroup for the adjoint action
of a  connected and simply connected compact Lie group on itself is connected.
\end{proof}

\begin{rem}\label{remG9}
The group $\SU(n)_{\delta(\xi)}^{u_0}$ contains the center $\bZ_n < \SU(n)$, and $\bZ_n$ gives the isotropy group
of every point of $\SU(n)_{\delta(\xi)}$ with respect to the action of $\SU(n)_{\delta(\xi)}^{u_0}$ given by
the formula \eqref{act33}. Indeed, this is a consequence of the fact that the action of
$\SU(n)_{\mu_0(y)}$ on $\mu^{-1}(\mu_0(y))$ is free except for the trivial action of the center $\bZ_n$
\end{rem}

Before turning to the characterization of the fiber $\hat \beta^{-1}(\xi)$, we need one more piece of preparation.
Let
$\su(n)_{\mu_0(y)}$ and $\k_\xi^{g_0}$ denote the Lie algebras of the groups $\SU(n)_{\mu_0(y)}$ and $K_\xi^{g_0}$, respectively.
For any $Z \in \su(n)_{\mu_0(y)}$ and $p = (A,B) \in \mu^{-1}(\mu_0(y))$ denote
\be
Z_p := ([Z,A], [Z,B]) \in T_p \mu^{-1}(\mu_0(y)).
\label{B12}\ee
This expression describes the vector fields induced by the action of $\SU(n)_{\mu_0(y)}$.

\begin{lem}\label{lem312}
If $Z_p$ belongs to $T_p M(\xi)\subset T_p \mu^{-1}(\mu_0(y))$ for
some $Z\in \su(n)_{\mu_0(y)}$ and $p\in M(\xi)$, then $Z \in \k_\xi^{g_0}$.
\end{lem}    \begin{proof}
If $Z_p$ belongs to $T_p M(\xi)$, then we must have $[Z,g_0^{-1} \delta(\xi) g_0]=0$, which means that
$Z = g_0^{-1} \Upsilon g_0$ for some element $\Upsilon \in \su(n)_{\delta(\xi)}$. But $\Upsilon$ must
also satisfy $[\Upsilon, \hat \mu(y,u_0)]=0$,
since $[Z,   \mu_0(y)] =0$. This means that $Z \in \k_\xi^{g_0}$.
\end{proof}

\begin{thm}\label{thm313}
Let us define the map
\be
F: M(\xi)/K_\xi^{g_0} \to  \mu^{-1}(\mu_0(y))/\SU(n)_{\mu_0(y)}
\label{B8}\ee
by the prescription
\be
F: K_\xi^{g_0} \cdot (A,g_0^{-1} \delta(\xi) g_0) \mapsto \SU(n)_{\mu_0(y)} \cdot (A,g_0^{-1}\delta(\xi) g_0),
\quad \forall (A,g_0^{-1} \delta(\xi) g_0)\in M(\xi).
\label{B9}\ee
This map is a smooth injective immersion whose image is $\hat \beta^{-1}(\xi)$.
\end{thm}
\begin{proof}
We begin by showing that every orbit of $\SU(n)_{\mu_0(y)}$ in $\mu^{-1}(\mu_0(y)) \cap \beta^{-1}(\xi)$ admits representatives
in $M(\xi)$, which means that the image of the map $F$ is $\hat \beta^{-1}(\xi)$.
We know that all elements  $(A,B)\in \beta^{-1}(\xi) \cap \mu^{-1}(\mu_0(y))$ are given by the formula
\eqref{G6}.
In this formula, according to Lemmas \ref{lemG3} and \ref{lemG4}, we can write
$u = X u_0$ with $X\in \U(n)_{\delta(\xi)}$ and substitute \eqref{G12} for $g$, which  gives
\be
(A,B) = (\eta g_0^{-1} X^{-1} \hat A_{Xu_0} T X g_0 \eta^{-1}, \eta g_0^{-1} \delta(\xi) g_0 \eta^{-1}).
\label{G18}\ee
One can check directly from the definitions that for any $X\in \U(n)_{\delta(\xi)}$ the element
\be
X':= \hat A_0^{-1}  X^{-1} \hat A_{Xu_0}
\ee
belongs to $\U(n)_{\delta(\xi)}$.
By using this, we can rewrite \eqref{G18} as
\be
(A,B) = (\eta g_0^{-1} \hat A_0 T' g_0 \eta^{-1},  \eta g_0^{-1} \delta(\xi) g_0 \eta^{-1})
\quad\hbox{with}\quad
T'= X' T X \in \SU(n)_{\delta(\xi)}.
\ee
As observed in  Remark \ref{remG5}, here $\eta$ can be replaced by $\tilde \eta \in \SU(n)_{\mu_0(y)}$.
Since $(\eta^{-1} A \eta, \eta^{-1} B \eta) \in M(\xi)$,
every $\SU(n)_{\mu_0(y)}$ orbit in $\mu^{-1}(\mu_0(y))\cap \beta^{-1}(\xi)$  has representatives  in $M(\xi)$.

To see that $F$ \eqref{B9} is injective, suppose that
\be
\SU(n)_{\mu_0(y)} \cdot (A_1,   g_0^{-1}\delta(\xi) g_0)= \SU(n)_{\mu_0(y)} \cdot (A_2,   g_0^{-1}\delta(\xi) g_0)
\label{B10}\ee
for two elements of $M(\xi)$. This means that there exists some $\eta \in \SU(n)_{\mu_0(y)}$ that transforms
$(A_1,   g_0^{-1}\delta(\xi) g_0)$ into $(A_2,   g_0^{-1}\delta(\xi) g_0)$. We have shown in Lemma \ref{lemG8}  that any such $\eta$ necessarily
belongs to $K_\xi^{g_0}$, whereby the $K_\xi^{g_0}$ orbits of $(A_1,   g_0^{-1}\delta(\xi) g_0)$ and $(A_2,   g_0^{-1}\delta(\xi) g_0)$ are the same.

To verify the smoothness of $F$ \eqref{B9}, notice that $M(\xi)$ is a locally trivial (left) principal fiber bundle
over $M(\xi)/K_\xi^{g_0}$, with structure group $K_\xi^{g_0}/\bZ_n$.
Thus, given any point $x_0 \in M(\xi)/K_\xi^{g_0}$, there exists a smooth local section
$\sigma: U_{x_0} \to M(\xi)$ over a coordinate neighbourhood  $U_{x_0}$ of $x_0$.
It is clear that $j_\xi \circ \sigma$ (see \eqref{B4}) is also smooth, and
the restricted map $F_{\vert U_{x_0}}$  can be written as
\be
F_{\vert U_{x_0}} = \pi_0 \circ j_\xi \circ \sigma,
\label{B11}\ee
where $\pi_0: \mu^{-1}(\mu_0(y)) \to \mu^{-1}(\mu_0(y))/\SU(n)_{\mu_0(y)}$ is the canonical projection.
Smoothness is a local property, and
$F_{\vert U_{x_0}}$ is obviously smooth.

It remains to show
that $F$ \eqref{B9} is an immersion, i.e., its derivative $DF(x_0)$ is injective at every point of $x_0\in M(\xi)/K_\xi^{g_0}$.
 Take an arbitrary non-zero tangent vector $X\in T_{x_0} (M(\xi)/K_\xi^{g_0})$. We have to show  that
$(DF(x_0)) (X)\neq 0$.
Choosing a local section $\sigma: U_{x_0} \to M(\xi)$ as in the previous step of the proof,
let us put $p:= \sigma(x_0)$.  Since $D j_\xi(p)$ is the trivial inclusion of $T_p M(\xi)$ into $T_p \mu^{-1}(\mu_0(y))$, by using \eqref{B11}
we can write
\be
(DF(x_0))(X)  = (D\pi_0(p)) ( (D\sigma(x_0))(X)).
\ee
The kernel of $D\pi_0(p)$ is the tangent space of the orbit $\SU(n)_{\mu_0(y)}\cdot p$ at $p$.
Therefore,  if $(DF(x_0))(X)=0$, then
\be
(D\sigma(x_0))(X) = Z_p
\quad\hbox{for some} \quad
Z \in \su(n)_{\mu_0(y)}.
\ee
Here, we used the notation \eqref{B12}.
But $(D\sigma(x_0))(X)$ belongs to $T_p M(\xi) < T_p \mu^{-1}(\mu_0(y))$.  Consequently,  by Lemma \ref{lem312},  $Z\in \k_\xi^{g_0}$ must hold,
which entails that $Z_p$ is tangent $K_\xi^{g_0} \cdot p$.
However, because $\sigma$ is a section, $(D\sigma(x_0))(X)$ is not tangent to the
$K_\xi^{g_0}$ orbit through $p$ for any non-zero $X$.
(Indeed, the section induces a local trivialization
$\pi_\xi^{-1}(U_{x_0}) \simeq U_{x_0} \times K_\xi^{g_0}/\bZ_n$,
and in this trivialization $\sigma$ becomes the map $U_{x_0} \ni x \mapsto (x, e)$, where $e$ is the unit element of the structure group.
Then, $D\sigma(x_0)(X)$ can be written as $(X, 0)$, while the $K_\xi^{g_0}$ orbit reads $(x_0, K_\xi^{g_0}/\bZ_n$.)
This proves that $(DF(x_0))(X) \neq 0$ for any non-zero tangent vector, i.e., $F$ is an immersion.
 \end{proof}

\begin{thm}\label{thm314}
The map $F$ \eqref{B8} is an embedding of the manifold $M(\xi)/K_\xi^{g_0}$ into $P(\mu_0(y))$.
Its image is the fiber $\hat \beta^{-1}(\xi)$,  which is therefore a connected embedded submanifold of $P(\mu_0(y))$.
\end{thm}
\begin{proof}
We have shown that $F$ is an injective immersion and its image is $\hat \beta^{-1}(\xi)$.
Hence, it is an embedding as $M(\xi)/K_\xi^{g_0}$ is compact.
 The connectedness follows from Lemma \ref{lem310}.
\end{proof}

We finish by showing that $\hat \beta^{-1}(\xi)$ is an \emph{isotropic} submanifold
of the symplectic manifold $(P(\mu_0(y)), \omega_{\mu_0(y)})$.
This will follow from the following lemma.

\begin{lem}\label{lem315}
The pull-back of the quasi-Hamiltonian 2-form $\omega$ \eqref{T3} vanishes on $M(\xi)$.
That is, with the tautological embedding $\iota_\xi: M(\xi) \to P$, we have
\be
\iota_\xi^* (\omega) =0.
\ee
\end{lem}
\begin{proof}
In practice, $\iota_\xi^*(\omega)$ is obtained by substituting $(A,B) \in M(\xi)$ into the
formula \eqref{T3}. Since in this case $B$ is constant, the first two terms vanish. By using that
$AB = \mu_0(y) BA$, we get
\be
\iota_\xi^*(\omega) = - \langle (\mu_0(y) BA)^{-1} \mathrm{d} (\mu_0(y) BA) \stackrel{\wedge}{,} (BA)^{-1} \mathrm{d} (BA)
\rangle = - \langle A^{-1} \mathrm{d} A \stackrel{\wedge}{,} A^{-1} \mathrm{d} A \rangle,
\label{G29}\ee
which vanishes obviously.
\end{proof}

\begin{thm}\label{thm316}
 The fiber $\hat \beta^{-1}(\xi)$ is an isotropic submanifold.
\end{thm}
\begin{proof}
First, we need to recall the key property of the symplectic form $\omega_{\mu_0(y)}$.
To do so,  let us fix an arbitrary point $z\in P(\mu_0(y))$ and two arbitrary tangent vectors $V, W \in T_z P(\mu_0(y))$.
Then, pick a point $\tilde z \in \mu^{-1}(\mu_0(y))$ and two vectors $\tilde V, \tilde W \in T_{\tilde z} \mu^{-1}(\mu_0(y))$
for which $z= \pi_0(\tilde z)$ and
\be
D\pi_0(\tilde z)(\tilde V) = V
\quad \hbox{and}\quad
D\pi_0(\tilde z)(\tilde W) = W.
\label{B16}\ee
The equality \eqref{T1} means that
\be
\omega_{\mu_0(y)} (V, W) = \omega(\tilde V, \tilde W).
\label{B17}\ee

Now, suppose that $z$, $V,W$ are such that $z\in \hat \beta^{-1}(\xi)$ and $V,W$ are tangent to $\hat \beta^{-1}(\xi)$.
In this case, we can choose $\tilde z\in M(\xi)$ and $\tilde V, \tilde W \in T_{\tilde z} M(\xi) \subset T_{\tilde z} \mu^{-1}(\mu_0(y))$
so that \eqref{B16} holds.
This is easily seen from the proof of the claim that $F$ \eqref{B8} is an immersion with its image being $\hat \beta^{-1}(\xi)$.
But then the $\omega(\tilde V, \tilde W)=0$ by Lemma \ref{lem315},  and hence the equality \eqref{B17} implies the claim.
\end{proof}

Theorems \ref{thm313}, \ref{thm314} and \ref{thm316} represent the main results of the present paper.
Their message is that the Liouville integrable systems of \cite{FK} exhibit `spherical singularities'
in all  type (ii) cases,
enriching the set of known integrable systems for which similar singularities were found previously \cite{B,BMZ,CK,CKO,K, NU-pol, NU}.
The rather concrete description
of the fibers, as given by Lemma \ref{lem310}, is also an important result.

\subsection{Recovering the torus structure of $\hat \beta^{-1}(\xi)$ for $\xi\in \cA_y \cap \cA^\reg$}\label{Sec32}

Although the structure of the fibers for $\xi \notin (\partial \cA_y)^\sing$ is known \cite{FK}, it appears worth
explaining how it arises as an easy consequence of Theorem \ref{thm314}, which is valid for any $\xi \in \cA_y$.
To do so, we shall apply the identification $\hat \beta^{-1}(\xi) = \SU(n)_{\delta(\xi)}/\SU(n)_{\delta(\xi)}^{u_0}$ provided
by Lemma \ref{lem310} and Theorem \ref{thm314}.

\begin{cor}\label{corG12}
Let us suppose that $\xi \in \cA_y \cap \cA^\reg$ and from the $n$ functions  $z_\ell$ ($\ell=1,\dots, n$)  in \eqref{E5}
$z_\ell(y,\xi)=0$  holds for a unique $d$-element subset, for some $0\leq d \leq (n-1)$.
Then, the fiber is a torus
\be
\hat \beta^{-1}(\xi) \simeq \bT^{n-1-d}.
\label{G32}\ee
\end{cor}
\begin{proof}
After choosing $\xi$, fix also a corresponding vector $u_0$.
It follows from \eqref{E4} that the number of zero components of $u_0$ is $d$.
Since $\xi\in \cA^\reg$, we have
\be
\SU(n)_{\delta(\xi)} = \bT^{n-1}.
\label{G33}\ee
As noted before, if in addition $\xi\in \cA_y^o$, then   $z_\ell(y,\xi)>0$ for all $\ell$, and thus $d=0$.
In this case, the action \eqref{act33} of
$\SU(n)_{\delta(\xi)}^{u_0} = \bZ_n$ is trivial, and thus
\be
\hat \beta^{-1}(\xi) \equiv \SU(n)_{\delta(\xi)}/ \SU(n)_{\delta(\xi)}^{u_0} = \bT^{n-1}, \qquad \forall \xi \in \cA_y^o,
\label{G34}\ee
which is a standard Liouville torus.

If $\xi \in (\partial \cA_y)^\reg$, then the vector $u_0$ has $1\leq d\leq (n-1)$ zero components.
This holds because  $z_\ell(y,\xi)$ vanishes precisely if one of the sums in \eqref{T24} is equal to $y$,
for the respective index $\ell$, and such equalities characterize the boundary $\partial \cA_y$.
In this case,
$\SU(n)_{\delta(\xi)}^{u_0}$ is a $d$-dimensional torus. For example, if the first $d$ components of $u_0$ are zero, then
$X\in  \SU(n)_{\delta(\xi)}^{u_0}$ has the form
\be
X = \diag(\tau_1, \dots, \tau_d, \tau,\dots, \tau) \quad \hbox{with}\quad  \tau_d = (\tau_1\cdots \tau_{d-1} \tau^{n-d})^{-1},
\label{G35}\ee
and thus $X\in \SU(n)_{u_0}$ is parametrized by $(\tau_1,\dots, \tau_{d-1}, \tau) \in \bT^d$.
Now, since $\SU(n)_{\delta(\xi)}$ is Abelian,
we can define the \emph{homomorphism}
\be
\varphi: \SU(n)_{\delta(\xi)}^{u_0} \to \SU(n)_{\delta(\xi)} = \bT^{n-1},
\quad
\varphi(X):= \Theta_0(X) X^{-1} = \hat A_0^{-1} X \hat A_0 X^{-1}.
 \label{G36}\ee
Then, the action (see \eqref{act33}) of $\SU(n)_{\delta(\xi)}^{u_0}$ on $\SU(n)_{\delta(\xi)}$ works according to the rule
 \be
 (X,T) \mapsto T \varphi(X),
 \label{G37}\ee
 and we know that this action is free, except for the trivial action of the center $\bZ_n < \SU(n)_{\delta(\xi)}^{u_0}$.
 By using this, the formula \eqref{G37} of the action  implies that the kernel of the homomorphism
 $\varphi$ is  $\bZ_n$ (which is not apparent from the formula \eqref{G36} of $\varphi$).
 It follows that the image
 \be
 \varphi(\SU(n)_{\delta(\xi)}^{u_0}) < \bT^{n-1}
 \label{G38}\ee
  is  a $d$-dimensional sub-torus (since all connected, compact, Abelian Lie groups are tori).
 As a result, the fiber gets identified with the coset space
 \be
 \hat\beta^{-1}(\xi) = \bT^{n-1}/  \varphi(\SU(n)_{\delta(\xi)}^{u_0}) \simeq \bT^{n-1-d}.
 \label{G39}\ee
 Here, we used that the quotient space (in fact, factor group) of a torus by a sub-torus is again a torus.
 \end{proof}

Recall that the momentum map $\hat \beta$ generates a Hamiltonian $\bT^{n-1}$ action on
\be
P(\mu_0(y))^\reg:= \hat \beta^{-1}(\cA_y \cap \cA^\reg),
\ee
which is a dense open submanifold of $P(\mu_0(y))$.
This action descends from the free action  \eqref{T16} on $ \beta^{-1}(\cA^\reg)$ and it remains
free  on $\hat \beta^{-1}(\cA_y^o)$, since $\SU(n)_{\delta(\xi)}^{u_0}=\bZ_n$ for $\xi\in \cA_y^o$ \eqref{G30}.

The restricted system  on $P(\mu_0(y))^\reg$ is quite similar to compact toric systems  \cite{A,D}.
For example, let us suppose that $\xi$ is a `regular vertex' of $\cA_y$, i.e., a vertex that belongs to $\cA_y\cap \cA^\reg$.
It follows from standard results on polytopes \cite{PV,Z} that $\xi$ must satisfy at least $(n-1)$ of the $2n$ defining relations  \eqref{T24}
as equalities.
Notice $\xi_\ell + \cdots + \xi_{\ell + k-1} =y$ and
$\xi_\ell + \cdots + \xi_k = y$ cannot hold together
since $\xi \in \cA^\reg$.
By \eqref{E5},  either of these equalities imply that $z_\ell(\xi,y) =0$.
Taking also into account that $\sum_{\ell =1}^n z_\ell(\xi,y) =1$,
we conclude that any regular vertex \emph{satisfies precisely $(n-1)$ of the defining relations as equalities}.
Since our polytope is of dimension $(n-1)$, it also follows that the pertinent $(n-1)$ equalities  correspond
to faces of co-dimension one.
This means that the regular vertices behave in the same way as the vertices of the so-called
\emph{simple polytopes} \cite{PV},  which contain the Delzant polytopes \cite{A,D} of toric systems as  a  subclass.

We see from the above discussion and Corollary \ref{G12} that $\hat \beta^{-1}(\xi)$ is a fixed point
of the torus action for any regular vertex $\xi$.
We remark that there is also a correspondence between different isotropy types of the Hamiltonian $\bT^{n-1}$
action and the intersections between  $\cA^\reg$ and faces of $\cA_y$ of different dimensions,
like in toric systems.
In particular, one can show that
 if  $\xi\in (\partial \cA_y)^\reg$ belongs to the relative interior of a face of co-dimension $d$, then
 the points of  $\hat \beta^{-1}(\xi)$ possess $d$-dimensional isotropy groups and are
 tori of dimension $(n-1-d)$ \eqref{G32}.

The characteristic property of the type (ii) cases is that $\hat \beta$ does not generate
a global torus action, since
$(\partial \cA_y)^\sing \neq \emptyset$.
The simplest examples of `singular fibers', i.e., fibers associated with $\xi\in (\partial \cA_y)^\sing$,
are presented in the next section.

\section{The type $\mathrm{(ii)}$ examples of lowest dimension}\label{Sec4}

We wish to apply the combination of  Lemma \ref{lem310} and Theorem \ref{thm314} for describing the
fiber $\hat \beta^{-1}(\xi)$ for $\xi \in (\partial \cA_y)^\sing$ \eqref{G31}.
Such  values of $\xi$  first occur for $n=4$, for any
\be
\frac{\pi}{3} < y < \frac{\pi}{2}.
\label{F1}\ee
The corresponding momentum polytope was already calculated\footnote{The polytope was  given in \cite{FK} without describing the calculation.
The result is confirmed by Theorem \ref{thm51} and Remark \ref{rem52} below, which include it as the $n=4$ special case.}
in \cite{FK}.
We next present it in detail, and then describe the singular fibers.

To avoid confusion, let us stress that the true momentum map is the map \eqref{T20}, but the $n$-th component
is often included for notational convenience. In our case, we write the points of the polytope as 4-component
vectors, but  Figure \ref{Fig1}  refers to their projection onto the 3-dimensional space of the first 3 components.

\begin{figure}
\includegraphics{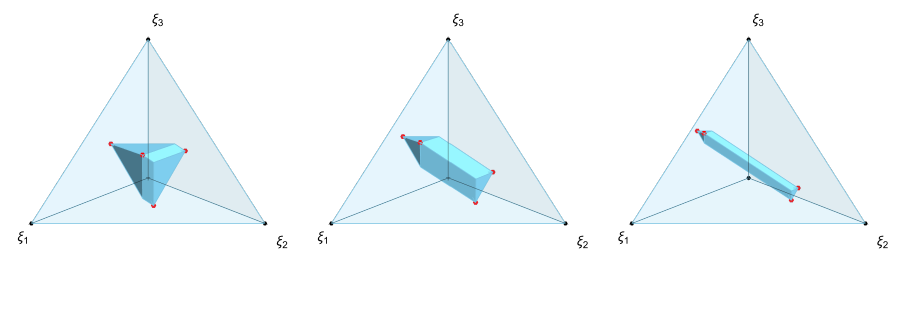}
\caption{The type (ii) polytope for $n=4$ at
 $y/\pi =  7/20, 5/12$ and $7/15$, from left to right. The red dots denote the `irregular' vertices.
 At the excluded values $y/\pi = 1/3$ and $1/2$ the polytope becomes a simplex and a line segment, respectively, with all vertices on $\partial \cA$,
 as seen from \eqref{Rvert} and \eqref{Ivert}.}
 \label{Fig1}
\end{figure}

\subsection{The $n=4$ momentum polytope}\label{Sec41}

Our 3-dimensional  momentum polytopes     have 8 vertices, 4
lying in $(\partial \cA_y)^\reg$ and another 4 in $(\partial \cA_y)^\sing$.
To display them, we employ
the cyclic permutation $\nu:= \sigma^{-1}$ \eqref{T25}.
The four regular vertices are
\be
R(1)= (y, \pi -2 y, 3y - \pi, \pi - 2y)
\quad \hbox{and}\quad  R(k) = \nu^{k-1}(R(1))
\quad
\hbox{for}\quad k=2,3,4,
\label{Rvert}\ee
and the four `singular' (irregular) vertices are
\be
I(1)= (y, \pi -2 y, y, 0)
\quad \hbox{and} \quad I(k) = \nu^{k-1}(I(1))
\quad
\hbox{for}\quad k=2,3,4.
\label{Ivert}\ee
It follows that
\be
(\partial \cA_y)^\sing = \{ I(1), I(2), I(3), I(4)\}.
\label{F3}\ee

The polytope also has 8 faces,  four are quadrilateral and another four are triangular
The quadrilateral face $\cQ(k)$ is defined by the equation $\xi_k + \xi_{k+1} = y$ (modulo 4).
With the above  labelling of the vertices we have
\be
\cQ(1) = \langle I(2), R(3), I(3), R(4) \rangle,
\ee
meaning that $\cQ(1)$ is the convex hull of the 4 listed vertices, which are listed in the order
of positive orientation as seen from the outside. From this, one also sees the edges that are incident with $\cQ(1)$.
The other quadrilateral faces are $\cQ(k) = \nu^{k-1}(\cQ(1))$ for $k=2,3,4$.

The triangular face $\cT(k)$ is defined by the equality $\xi_k =y$.  For example, we have
\be
\cT(1) = \langle I(1), R(1), I(3) \rangle
\quad \hbox{and} \quad \cT(k) = \nu^{k-1}(\cT(1))
\quad
\hbox{for}\quad k=2,3,4,
\ee
with the vertices listed anti-clockwise.
It should be noted that the application of the cyclic permutation changes the orientation.
For example,  the vertices of $\cT(2) = \langle I(2), R(2), I(4) \rangle $
are listed clockwise, and similarly for $\cQ(2)$.

Every irregular vertex is incident with 4 faces, 2 triangular and 2 quadrilateral, and with 4 edges.
For example, we have
\be
I(1) =  \cT(1) \cap \cT(3) \cap \cQ(3) \cap \cQ(4)
\ee
and
\be
I(1) = \langle I(1), R(1) \rangle \cap \langle I(1), R(2) \rangle \cap \langle I(1), R(3) \rangle \cap \langle I(1), I(3) \rangle.
\ee

Every regular vertex is incident with 3 faces and with 3 edges.
For example, we have
\be
R(1) =  \cT(1) \cap \cQ(2) \cap \cQ(3)
\ee
and
\be
R(1) = \langle R(1), I(1) \rangle   \cap \langle R(1), I(3) \rangle \cap \langle R(1), I(4) \rangle .
\ee
The other incidence relations follow by the application of the cyclic permutation $\nu$.

\begin{rem}\label{rem41}
One can check that the 3 edges emanating from any regular vertex satisfy the conditions
of a Delzant polytope. Namely, they are parallel to 3 vectors that have integer components
and form a basis of the lattice $\bZ^3$. Equivalently, the normal vectors of the 3 faces incident with a regular
vertex can be  normalized so that they form a basis of $\bZ^3$.
\end{rem}

\begin{rem}\label{rem42}
The above remark refers to the projection to the space of the first 3 components,
but actually the full symmetry of the geometric object $\cA_y$ is better seen from
the perspective of the Euclidean hyperlane $E \subset \bR^4$ \eqref{T23}.
There, the following `reconstruction' is available.
First, replace the 4 vertices $R(k)$ by new vertices $R(k)'$ given by
\be
R(1)' := \diag( \frac{\pi}{2}, \frac{\pi}{2}-y, 2y - \frac{\pi}{2}, \frac{\pi}{2}-y),
\ee
and its cyclic permutations $R(k)' = \nu^{k-1}(R(1))$ for $k=2,3,4$.
With respect to the distance in $E$, this yields a rectangular parallelepiped, whose faces
\be
\langle I(1), R(1)', I(3), R(3)'\rangle
\quad\hbox{and}\quad
 \langle R(2)', I(4), R(4)', I(2) \rangle
 \ee
are squares of edge length $(\pi - 2y)$ and the other two edges of
$ Q(1)' =  \langle I(2), R(3)', I(3), R(4)' \rangle$ and its permutations
have length $(4 y - \pi )$. Then, $\cA_y$ is reconstructed by chopping off the new vertices along
opposite diagonals of the squares, in such a way that results in equilateral triangles $\cT(k)$ in $E$.
\end{rem}

\subsection{The $n=4$ singular fibers}\label{Sec42}

We below describe the structure of $\hat \beta^{-1}(\xi)$  for $\xi= I(2)$ given by
\be
\xi = (0,y, \pi -2y, y)
\quad \hbox{for which} \quad
 \delta(\xi)  = - \ri\, \diag( 1,1, e^{2\ri y},  e^{-2\ri y}).
\label{F4}\ee
The other 3 singular fibers have the same structure.

Now we have
\be
\SU(4)_{\delta(\xi)} = \{T= \diag (\Gamma, \Gamma_3, \Gamma_4) \mid \Gamma \in \U(2),\, \Gamma_3,
\Gamma_4\in \U(1),\, \Gamma_3 \Gamma_4 \det(\Gamma) =1 \} \simeq \U(2) \times \U(1).
\label{F5}\ee
By inspecting the equality \eqref{E3}, we find that the corresponding  unit vectors $u\in \bC^4$ are constrained  by
\be
 \quad \vert u_1 \vert^2 + \vert u_2 \vert^2 = - 1/(2 \cos(2y)), \,\,
\vert u_3\vert^2 = 1 + 1/(2 \cos(2y)), \,\,  u_4 =0.
\label{F6}\ee
For the application of Lemma \ref{lem310}, we  choose the vector $u_0$ as
\be
u_0 = (0, r, \sqrt{1-r^2}, 0)^T \quad \hbox{with}\quad r: =  \sqrt{- 1/(2\cos(2y))}.
\label{F7}\ee
Note that $r$ and $\sqrt{1-r^2}$ are positive.
Then, the group $\SU(u)_\xi^{u_0} < \SU(4)_{\delta(\xi)}$  can be displayed as
\be
\SU(4)_\xi^{u_0} =\{ X= \diag(a, b, b, c) \mid a,b,c \in \U(1),\,  ab^2c=1\} \simeq \U(1) \times \U(1).
\label{F8}\ee
Next, we have to fix  $\hat A_0 \in \SU(4)$ that satisfies $\hat A_0 \delta(\xi)  =  \hat \mu(y, u_0) \hat A_0$.
With $u_0$ \eqref{F7}, we can write
\be
\hat \mu(y, u_0) \delta(\xi)  = -\ri \,\diag( e^{2\ri y}, e^{2\ri y}, e^{4\ri y}, 1) - \diag(0,Q, 0),
\label{F9}\ee
with a  2 by 2 matrix  $Q$ that we do not need to spell out.
A brief inspection of the equation
\be
  \hat A_0  \diag(\ri,\ri, \ri e^{2\ri y}, \ri e^{-2\ri y})   = \left( \diag(\ri e^{2\ri y},
  \ri e^{2\ri y}, \ri e^{4\ri y}, \ri )   + \diag(0, Q, 0)\right)  \hat A_0
 \label{F10}\ee
  shows that we can a find a solution of the  following form:
\be
\hat A_0 = \begin{bmatrix}
0&0&1&0\\
0& x_1 & 0& x_3 \\
0& x_2 & 0& x_4 \\
1 & 0&0&0
\end{bmatrix} .
\label{F11}\ee
For $X\in \SU(4)_\xi^{u_0}$ \eqref{F8}, this leads to
\be
\Theta_0(X) = \hat A_0^{-1} X \hat A_0 = \diag(c,b,a,b).
\label{F12}\ee
According to Theorem \ref{thm314} and Lemma \ref{lem310}, the fiber can be identified with the orbit space for the
action of $\SU(4)_\xi^{u_0}$ on $\SU(4)_{\delta(\xi)}$   defined by
\be
(X, T) \mapsto \Theta_0(X) T X^{-1}.
\label{F13}\ee
Taking $T$ from \eqref{F5} and $X$ from \eqref{F8}, the right-hand side reads
\be
\Theta_0(X) T X^{-1}  = \diag(c,b, a,b) \diag(\Gamma, \Gamma_3, \Gamma_4) \diag(a^{-1}, b^{-1}, b^{-1}, c^{-1}).
\label{F14} \ee
If $\Theta_0(X) T X^{-1} = T$, then $a=b=c$, which confirms that the isotropy group of every point is the center $\bZ_4< \SU(4)$ with respect to this action.

\begin{prop}\label{prop42}
Every orbit of $\SU(4)_\xi^{u_0}$ on $\SU(4)_{\delta(\xi)}$ has a unique representative of the form
$\diag(Z,1,1)$ with $Z\in \SU(2)$, and therefore, as a manifold, the singular fiber over $\xi$ \eqref{F4}  admits the identification
\be
\hat \beta^{-1}(\xi) \simeq S^3.
\ee
\end{prop}
\begin{proof}
By choosing $X=\diag(a,b,b,c)\in \SU(4)_\xi^{u_0}$ so that
\be
\Gamma_3 a/b= \Gamma_4 b/c = 1
\ee
we  can transform any $T=(\Gamma, \Gamma_3, \Gamma_4)$ into an element of the form $\diag(Z,1,1)$.
This requires
\be
b^4 = \Gamma_3/\Gamma_4,\,\, a = b/\Gamma_3,\,\, c = b\Gamma_4.
\ee
Thus, for a given $T$,  $X$ is unique up to multiplication by an element of $\bZ_4$ and $Z$ is unique.
Consequently,  the smooth manifold $\SU(4)_{\delta(\xi)}/\SU(4)_\xi^{u_0}$ is
 diffeomorphic to the global cross-section of the orbits given by
 $\{ \diag(Z,1,1)\mid Z\in \SU(2)\}$.
We better think of the fiber as $S^3$, since no natural group structure arises on the quotient space $\SU(4)_{\delta(\xi)}/\SU(4)_\xi^{u_0}$.
\end{proof}

\begin{rem}\label{remSeid}
Combining Theorems \ref{thm314} and \ref{thm316} with Proposition \ref{prop42} we see that the six dimensional compact
symplectic manifold $P(\mu_0(y))$  \eqref{T18} associated to our type (ii) examples for $n=4$ contains embedded Lagrangian $S^3$ spheres.
This implies by Theorem 3.1a of \cite{S} that these manifolds cannot be symplectomorphic to $\bC P^3$ with
a multiple of the Fubini--Study symplectic
form that appears \cite{FKlim,FK} in the type (i) cases for $\SU(4)$.
\end{rem}

\subsection{Dynamics on the $S^3$ fibers }\label{Sec43}

The phase space $P(\mu_0(y))$ carries the Abelian Poisson algebra $\hat \fH_b$ that descends from
$\pi_2^*C^\infty(\SU(n))^{\SU(n)}$. The joint level surfaces of $\hat \fH_b$ are the same as the
fibers of the momentum map $\hat \beta$. Below, we describe the dynamics generated by  these smooth Hamiltonians
on the singular $S^3$ fiber.

For $g\in \SU(n)$, let $\fc_g$ denote the span of the derivatives $\nabla h(g)$ of the conjugation invariant smooth functions.
It is well known that $\fc_g$ is  the center of the Lie algebra of the isotropy group  $\SU(n)_g$.
For the singular vertex $\xi$ \eqref{F4}, $\fc_{\delta(\xi)}$ is two dimensional.
A basis of  $\fc_{\delta(\xi)}$ is provided by
\be
X_1:= -\ri\, \diag(1,1,-1,-1) \quad\hbox{and}\quad X_2 := -\ri\, \diag(0,0,1,-1).
\label{425}\ee

Consider the following invariant functions on $\SU(n)$,
\be
\phi_k(g):= \frac{1}{2k} \tr(g^k + g^{-k})
\quad\hbox{and}\quad
 \phi_{-k}(g):= \frac{1}{2\ri k} \tr(g^k - g^{-k}),
 \quad \forall k\in \bZ_+,
 \ee
 whose derivatives are
 \be
 \nabla \phi_k(g) = (g^{-k} - g^k) + \frac{\1_n}{n} \tr (g^k - g^{-k}),
 \qquad
\nabla \phi_{-k}(g) = \ri (g^{k} + g^{-k}) - \frac{\ri \1_n}{n} \tr (g^k + g^{-k}).
\ee
By using this, one can easily find smooth invariant functions on $\SU(4)$ whose
derivatives at $\delta(\xi)$ \eqref{F4}  are $X_1$ and $X_2$ \eqref{425}.

\begin{prop}\label{prop44}
Consider two functions $h_j \in C^\infty(\SU(4))^{\SU(4)}$ for which
$\nabla h_j(\delta(\xi)) = X_j$ \eqref{425} with $\xi$ \eqref{F4}, and let $\cH_j\in \hat \fH_b$
denote the corresponding reduced Hamiltonians defined by
\be
h_j \circ \pi_2 \circ \iota_0 = \cH_j \circ \pi_0, \quad j=1,2.
\ee
Following the proof of Proposition \ref{prop42},
identify the fiber $\hat \beta^{-1}(\xi)$ with the submanifold $M_0(\xi)$ of $M(\xi)$ given by
\be
M_0(\xi): = \{ (g_0^{-1} \hat A_0 \,\diag(Z,1,1) g_0, g_0^{-1} \delta(\xi) g_0) \mid Z\in \SU(2)\}.
\ee
Then,  the integral curves of $\cH_1$ that belong to the fiber can be written as
\be
Z_1(t) = \diag( e^{\ri t}, 1)\, Z_0\,
  \diag( 1, e^{-\ri t})
\label{Z1t}\ee
and the integral curves generated by $\cH_2$ can be written as
\be
Z_2(t) = \diag( e^{\ri t}, 1)\, Z_0\, \diag( e^{-\ri t}, 1),
\label{Z2t}\ee
where $Z_0 \in \SU(2)$ is an arbitrary initial value.
Parametrizing $Z \in \SU(2)$ as
\be
Z= \begin{bmatrix} (q_1 + \ri q_2) & (q_3 + \ri q_4) \\
(-q_3 + \ri q_4) & (q_1- \ri q_2)
\end{bmatrix}\,,
\ee
these integral curves describe rotations in the $q_1, q_2$ and in the $q_3, q_4$ planes, respectively.
\end{prop}
\begin{proof}
This is a straightforward calculation. By applying \eqref{T7}, one first writes down the integral curve of $\cH_j$ starting at an initial value
in $M_0(\xi)$. This integral curve stays in $M(\xi)$. Then, one projects it back onto $M_0(\xi)$ by using the
action on $\SU(4)_{\delta(\xi)}^{u_0}$, as in the proof of Proposition \ref{prop42}.
\end{proof}

\subsection{Conjecture on a semi-local isomorphism}\label{Sec44}

The $2\pi$-periodic integral curves exhibited in Proposition \ref{prop44} give rise to a $\bT^2$ action on the fiber
$\hat \beta^{-1}(\xi)$ with $\xi$ in \eqref{F4}, and one may ask if this is also generated  by the action variables.
In fact, it can be shown that although $\hat \beta_1$ is not differentiable at $\hat \beta^{-1}(\xi)$  where it vanishes,
the following 3 functions are $C^\infty$ in an open neighbourhood of $\hat \beta^{-1}(\xi)$,
\be
\hat \beta_1^2,\,\,   (\hat \beta_1 + 2 \hat \beta_2 + \hat \beta_3),\,\, \hat \beta_3,
\label{3smooth}\ee
and the latter two action variables have the same integral curves  on the fiber as given by equations \eqref{Z1t} and \eqref{Z2t}.
The proof of this statement will be presented elsewhere.

The cotangent bundle $T^*S^3$ carries a Liouville integrable system whose restriction to a neighbourhood of the zero section
is strikingly similar to the restriction of our system on a neighbourhood of the spherical fiber $S^3\simeq \hat \beta^{-1}(\xi) \subset P(\mu_0(y))$.
The system on   $T^*S^3$, equipped with the canonical symplectic form, is defined
 by the  triple of Poisson commuting Hamiltonians
\be
H,\, L_{12},\, L_{34},
\label{L1234}\ee
where $L_{ij}$ denotes the angular momentum  component generating rotations in the $ij$ plane, and $H=\sum_{1\leq i<j \leq 4} L_{ij}^2$ is
twice the kinetic energy associated with the round metric on  $S^3 \subset \bR^4$.
The angular momenta $L_{12}$ and $L_{34}$ are globally smooth and generate $2\pi$-periodic flows.
The integral curves of $\sqrt{H}$  outside the zero section are also $2\pi$-periodic,
but this function is not differentiable at the zero section.
This is very similar to the behaviour of the triple \eqref{3smooth} in a neighbourhood of the singular fiber.
To exhibit an even stronger resemblance, let us consider the map
\be
W: T^* S^3 \to \bR^3, \quad
W = (W_1, W_2, W_3):= (\sqrt{H}, L_{12}, L_{34}).
\label{H+6}\ee
The range of this map \cite{K} is the polyhedron in $\bR^3$ whose points $(w_1,w_2,w_3)$ satisfy
\be
w_1 \geq \vert w_2\vert + \vert w_3 \vert.
\label{H+7}\ee
It has a vertex at the origin and 4 edges emanate from it along the directions  of the vectors
\be
v_1 = (1,1,0),\,\, v_2 = (1,0,1),\,\, v_3 = (1,-1,0),\,\, v_4 = (1,0,-1).
\label{H+4}\ee
Consider also the map
\be
X: P(\mu_0(y)) \to \bR^3,
\quad
X= (X_1, X_2, X_3):= (\hat \beta_1, \hat \beta_1 + 2 \hat \beta_2 + \hat \beta_3 - \pi, \hat \beta_3 - (\pi - 2y)).
\label{H+1}\ee
The  resulting  image of $P(\mu_0(y)$  can be obtained from our polytope on Figure \ref{Fig1}.
The image of $I(2)$ is the origin, and the direction vectors of
 the images of the 4 edges
\be
\langle I(2), R(2) \rangle,\,\, \langle I(2), R(3) \rangle,\,\, \langle I(2), R(4) \rangle,\,\, \langle I(2), I(4) \rangle
\label{H+5}\ee
are the same vectors \eqref{H+4} that arose from the `geodesic system' \eqref{H+6}.

Although the elements of both triples $W$ and $X$ have $2\pi$-periodic integral curves, these are not  true action maps
because the integral curves of the Hamiltonian
$(W_1 + W_2 + W_3)$ and of the Hamiltonian $(X_1 + X_2 + X_3)$ both have period $\pi$.
In the latter case this is clear from the equality
\be
X_1 + X_2 + X_3 = 2( \hat \beta_1 + \hat \beta_2 + \hat \beta_3) - 2(\pi - y),
\label{H+2}\ee
and concerning  $W$ \eqref{H+6} it has been shown in \cite{McLH}.

For our system, instead $(\hat \beta_1, \hat \beta_2, \hat \beta_3)$ we may introduce alternative action variables as
the components of the map
\be
\widetilde{X}:= ((X_1 + X_2 + X_3)/2, X_2, X_3) = (\hat \beta_1 + \hat \beta_2 + \hat \beta_3 - (\pi - y),
\hat \beta_1 + 2 \hat \beta_2 + \hat \beta_3 - \pi, \hat \beta_3 - (\pi - 2y)).
\label{W+8}\ee
Notice that  the change of action variables $(\hat \beta_1, \hat \beta_2, \hat \beta_3) \mapsto \widetilde X$
is an integral affine transformation, i.e., its linear part is given by a matrix in $\mathrm{SL}(3,\bZ)$;
while the linear part of
the affine map  $(\hat \beta_1, \hat \beta_2, \hat \beta_3) \mapsto X$ contains a matrix of determinant two.
The obvious analogue of $\widetilde{X}$ is
\be
\widetilde{W}:= ((W_1 + W_2 + W_3)/2, W_2, W_3) = ((\sqrt{H} + L_{12} + L_{34})/2, L_{12}, L_{34}).
\ee

According to a theorem of Weinstein \cite{W}, there exist open neighbourhoods of the zero section in  the cotangent bundle $T^*S^3$
and of the Lagrangian
submanifold $S^3\simeq \hat \beta^{-1}(I(2)) \subset P(\mu_0(y))$ that are symplectomorphic by a map  sending   the zero section onto $\hat \beta^{-1}(I(2))$.
We believe that a symplectomorphism can be chosen in such a way to intertwine the two integrable systems.

\begin{conj}
There exists a symplectomorphism  $\Phi: U \to \widehat U$  between open
 neighbourhoods $U$ and $\widehat U$ of the zero section of $T^*S^3$ and of the fiber $\hat \beta^{-1}(I(2))$  in $P(\mu_0(y))$
that satisfies the relation
\be
X \circ \Phi = W,
\label{H+conj1}\ee
 with $X$ \eqref{H+1} and $W$ \eqref{H+6}.
As a consequence, $\Phi$ also satisfies
\be
\hat \beta_1^2 \circ \Phi = H
\quad\hbox{and}\quad
 \widetilde{X} \circ \Phi = \widetilde{W}.
 \label{H+conj2}\ee
\end{conj}

 Notice that $\Phi$ subject to \eqref{H+conj2} must map
the zero section onto $\hat \beta^{-1}(I(2))$, since they are the zero sets of $H$ and $\hat \beta_1$, respectively.
The validity of the conjecture would also imply that $\widetilde{W}$ is a true action map outside the zero section,
generating a free $\bT^3$ action on its generic fibers.

It is known \cite{Al,BMZ} that the `geodesic system' \eqref{L1234} is semi-locally equivalent to the
 Gelfand-Cetlin system on a generic coadjoint orbit of the unitary group ${\mathrm{U}}(3)$, whose momentum map also admits
 an $S^3$ fiber (first found in \cite{NNU}).
 In the future, we hope to prove the conjecture by adapting the arguments of \cite{Al} to our case.

\section{Further examples}\label{Sec5}

Here, we generalize
the results of the previous section for the simplest type (ii) cases that occur for any $n\geq 4$.
These examples arise by taking the parameter $y$ as small as possible from the possible ranges of type (ii) values.
It turns that the polytope $\cA_y$ has $(n-2)n$ vertices.  We display the vertices  explicitly and also prove that all
the $2n$ defining inequalities correspond to facets.
Out of the vertices $n$ are regular and $n(n-3)$ are singular, with each of the latter having a single zero component.
By proceeding essentially as in Subsection 4.2, we  demonstrate that the fiber over every singular vertex
is a copy of $S^3$.
 We describe these results in the next subsection, and then present
 miscellaneous  further results.

 Let us a recall a few standard facts \cite{Br,PV,Z} about convex $D$-dimensional  polytopes (bounded polyhedra)
 defined by linear inequalities in an affine space of
 dimension $D$.
 A linear inequality defines a half-space, with equality holding on its boundary hyperplane.
  A proper face of a polytope $\cP$  is a non-empty intersection of $\cP$ with a hyperplane
 that bounds a half-space containing  $\cP$.
 A vertex is a face consisting of a single point.
  The dimension of a face is the dimension of its affine span.
 The polytope is the convex hull of its vertices.
 The proper faces are themselves polytopes, with their vertices giving subsets of the vertices of $\cP$.
 (The full polytope and the emptyset are non-proper faces, of dimension $D$ and $-1$, respectively.)
 The $(D-1)$-dimensional faces are called facets, and the corresponding hyperplanes form in general a subset of
 those that occur in the definition of the polytope.
 The intersection of any family of faces is a face, and all proper faces can be obtained as intersections of subsets of the facets.
 A $d$-dimensional face is contained in at least $(D-d)$ facets. In particular, a vertex is a zero dimensional face,
 and  thus it must be incident with at least $D$ facets.
 If a subset of the defining hyperplanes intersects $\cP$ in a single point, then this point
 is a vertex, as it is a face containing only that point.

\subsection{The simplest type (ii) cases for any $n$}\label{Sec51}
 We consider those type (ii) cases  for $\SU(n)$ whose parameter $y$ satisfies the condition
\be
 \frac{\pi}{n-1} < y < \frac{\pi}{n-2},
\label{51}\ee
for an any $n\geq 4$.
 We note that $1/n$ and $1/(n-1)$ as well as $1/(n-1)$ and $1/(n-2)$ are neighbours in the Farey sequence $\fF_n$.
Thus, we see from Proposition  \ref{thmI1} that the condition \eqref{51} restricts $y$ to the smallest type (ii) admissible values
 within the range \eqref{T21} with $k=1$.
 In general, there occur several other type (ii) intervals in this range.

\begin{thm}\label{thm51}
The polytope $\cA_y \subset E$ \eqref{T23} associated with $\SU(n)$ $(n\geq 4)$ and $y$ in the range \eqref{51} has $n(n-2)$ vertices and $2n$ facets.
The vertices are cyclic permutations of a regular vertex, called $R$, and $(n-3)$ singular vertices,
denoted $I_s$, which are given explicitly as follows:
\be
R:\, \xi_1=  (n-1)y - \pi,\,\, \xi_2 = \xi_n = \pi-(n-2)y,\,\,  \xi_j=y \,\,\, \hbox{for all other components},
\label{54}\ee
and,  with $s=1,\dots, n-3$,
\be
I_s:\, \xi_1=0,\, \,\xi_{2 + s} =\pi -(n-2) y,\,\, \xi_j=y \,\,\, \hbox{for all other components}.
\label{52}\ee
Each of the $2n$ defining equations
\be
\xi_\ell \leq y, \quad \xi_\ell + \xi_{\ell +1} \geq y,\quad  \ell=1,\dots, n,
\label{ineqs}\ee
corresponds to a facet.
 Every regular vertex is incident with $n-1$ facets, and every  singular one is incident with $n$ facets.
 \end{thm}

 \begin{proof}
The first part of the proof justifies the claim about the vertices and the second one deals with the facets.

 We begin by showing that all vertices are necessarily given be formulas \eqref{54}, \eqref{52} and their repeated cyclic permutations.
 Supposing that $\xi$ is a vertex, it must satisfy all the defining inequalities \eqref{ineqs}, and at least
 $(n-1)$ of them must hold as equalities.
 A key point to observe is that $\xi$ cannot satisfy two equalities
 \be
 \xi_a + \xi_{a+1} = y \quad \hbox{and}\quad \xi_{b}+ \xi_{b+1}=y
 \ee
with four different subscripts (understood modulo $n$), since otherwise the sum of the components of $\xi$ can be at most
$2y + (n-4)y = (n-2) y$, which is smaller than $\pi$ (because of \eqref{51}).
 It also cannot be that only one equality of the
type $\xi_\ell + \xi_{\ell +1} =y$ holds, since in that case we could write at least $(n-1)$  equalities for $\xi$ only by setting
all other $(n-2)$ components equal to $y$, but this is excluded since $(n-1) y > \pi$ by \eqref{51}.
 Therefore, up to cyclic permutations, every vertex must have the form
 \be
 \xi = (A, B,\xi_3,\dots, \xi_{n-1}, B),
 \label{ansatz} \ee
  where $\xi_1 + \xi_2 = \xi_n + \xi_1= A + B =y$.
  Notice that $\xi \in \cA_y$ cannot have two zero components, and in particular $B=0$ cannot hold, since $(n-2)y < \pi$.

 If $A=0$, then $B =y$. Out of the other $(n-3)$ components, only one can be smaller than $y$.
 Indeed, only 3 are independent from the 4 equalities $\xi_2=y$,  $\xi_n=y$, $\xi_1 + \xi_2=y$, $\xi_n + \xi_1=y$,
 and thus we need $(n-4)$ further equalities of the form $\xi_j=y$ in order to obtain a vertex.
 (As is  easily checked directly, with a smaller number of such equalities one obtains a face of non-zero dimension.)
 Supposing that the `special component' is $\xi_{2+s} <y$ for some $1\leq s \leq (n-3)$, the
 condition $\pi = \xi_1+ \dots + \xi_n  = (n-2) y + \xi_{2+s}$  leads to $\xi = I_s$ \eqref{52}.
 This point belongs to $\cA_y$ since  $0<\pi - (n-2)y <y$ due to \eqref{51}.

Let us summarize what we have proved up to now.
 The point $\xi=I_s\in \cA_y$ satisfies
\be
\xi_j = y\quad \hbox{for}\quad j \in \{1, 2,\dots,n\} \setminus \{1, 2+s\}
\quad\hbox{and}\quad \xi_1 + \xi_2 = y, \quad \xi_n + \xi_1 =y.
\label{53}\ee
The last equation is a consequence of the previous $(n-1)$ equations,
which are obviously independent. Our considerations imply
that the face satisfying these equalities consists of the single point $I_s$, and hence this is a vertex.
We have also seen that the vertices obeying \eqref{ansatz} with $A=0$ are exhausted by the
points $I_s$ with $s=1,\dots, n-3$.
A  further straightforward inspection shows that the cyclic permutations of the vertices $I_s$ are all different.
For a fixed $s$, they are different since the $0$ is in different positions.
For different values of $s$, they are different since the $0$ is followed by a different
number of components that are equal to $y$, before reaching the special component $\xi_{2+s}$.

 To get the full set of  vertices, we return to the formula \eqref{ansatz}, but now assume that $ A>0$.
 (Remember that  $\xi_\ell \geq 0$, $\forall \ell$ follows from  \eqref{54}.)
 In this case, all the other $(n-3)$ components of $\xi$ must be equal to $y$, otherwise $\xi$ cannot satisfy $(n-1)$ of the
 $2n$ defining  conditions \eqref{ineqs} as equalities. Then, the two equations $A + B =y$ and
 $\xi_1 + \dots + \xi_n= A + 2B + (n-3) y = \pi$  lead to $\xi=R$ \eqref{54}.
 On account of \eqref{51}, we see that $R \in \cA_y$.
 This is a vertex since
 it lies on the intersection of $(n-1)$ independent hyperplanes defined by
the equations:
\be
\xi_j = y \quad \hbox{for}\quad j \in \{3, \dots, n-1\}
\quad\hbox{and}\quad \xi_{1} + \xi_{2} = y, \quad \xi_{n} + \xi_1 =y.
\label{55}\ee
(One way to see
the independence of the hyperplanes is to view them in $\bR^n$ and verify that their normal vectors together with the
normal vector of $E$ \eqref{T23} form a basis of $\bR^n$.)
The proof also shows directly that $\xi=R$ is the unique point of $\cA_y$ satisfying \eqref{55}.
The cyclic permutations of $R$ give $n$ vertices.

Turning to the second part of the proof, we observe
from the  list of vertices that all the $2n$ hyperplanes of $E$ \eqref{T23}, defined by
\be
\xi_\ell = y,\quad \xi_\ell + \xi_{\ell+1}=y, \quad \ell=1,\dots, n,
\label{56}\ee
intersect the polytope $\cA_y$.
It remains to prove  that each of these intersections is a facet.

We  first demonstrate that the face defined by the intersection of $\cA_y$ with the hyperplane   $\xi_n=y$ of $E\subset \bR^n$
\eqref{T23}  is a facet.
To this end, we verify that the vertices satisfying $\xi_n=y$ generate an affine space of dimension $(n-2)$.
 Let us put
\be
a_n:= (n-1) y - \pi, \qquad b_n:= \pi - (n-2)y,
\ee
and denote by $e_i$ ($i=1,\dots,n$) the standard basis of $\bR^n$, now written as row vectors.
Consider the following pairs of vertices on the face $\xi_n=y$:
\be
(0,y,b_n,\dots, y)
\quad\hbox{and}\quad
(b_n,a_n,b_n,\dots, y),
\ee
\be
(y,0, y, b_n, \dots , y)
\quad\hbox{and}\quad
(y,b_n,a_n,b_n, \dots, y),
\ee
and $(n-4)$ pairs of the form
\be
(0,y, \dots,b_n,y,\dots, y)
\quad\hbox{and}\quad
(0,y,\dots, y,b_n,\dots, y).
\ee
The dots are in the same positions for each pairs and are all filled with $y$.
The differences of the displayed  pairs of vertices give non-zero multiples of the vectors
$(e_1-e_2), (e_2-e_3)$, and $(e_j - e_{j+1})$ for $3 \leq   j \leq n-2$.
This implies that the affine span of the vertices in the face $\xi_n=y$ is of dimension $(n-2)$,
proving that this face is a facet. By cyclic permutations, all faces defined
by $\xi_j=y$ are facets.

To finish, consider the face $\xi_1 + \xi_2 =y$ and its $(n-2)$  pairs of vertices of the form
\be
(0,y,b_n,\dots,y)
\quad\hbox{and}\quad
(b_n,a_n, b_n,\dots, y),
\ee
\be
(0,y,\dots, b_n,y, \dots )
\quad\hbox{and}\quad
(y,0,\dots, y,b_n,\dots).
\ee
The dots are missing if $n=4$, otherwise they are filled with $y$.
The corresponding differences give $b_n(e_1-e_2)$ and
\be
y(e_1-e_2) + a_n (e_j - e_{j+1}),\quad  j=3,\dots, n-1.
\ee
The dimension of the span of these vectors is clearly $(n-2)$.
Using also the cyclic permutation symmetry, this proves that every face having a defining
equation $(\xi_\ell + \xi_{\ell +1}) =y$ is a facet.
\end{proof}

\begin{rem}\label{rem52}
Let us display the full list of vertices contained in a facet.
For this, let $\sigma^k$ ($k=0,1, \dots, n-1$) denote the powers of the cyclic permutation \eqref{T25}, $\sigma(\xi)_j = \xi_{j+1}$.
One easily sees that $\xi_1=y$ holds for the following vertices:
\be
\sigma^k(R), \quad k=2,\dots, n-2
\quad \hbox{and}\quad \sigma^k(I_s), \quad k\in \{ 1,\dots, n-1\}\setminus \{2s+1\},
\label{57}\ee
for all $s=1,\dots, (n-3)$,
and $\xi_1+\xi_2 =y$ holds for the following vertices:
\be
R,\, \nu(R),\, I_s,\, \nu(I_s),
\label{59}\ee
where $\nu = \sigma^{-1}$.
Therefore the facet $\xi_1=y$ contains
\be
(n-3) + (n-3)(n-2) = (n-3)(n-1)
\label{58}\ee
vertices and the facet $\xi_1+\xi_2 =y$ contains
$2(n-2)$ vertices.
With some effort, since we have all vertices of every facet, the list of ridges
(codimension 2 intersections of pairs of facets) and other lower dimensional faces can also be derived from our results.
\end{rem}

Next, we deal with the fibers over the singular vertices.
We need a preliminary lemma.

\begin{lem}\label{lemfs}
Let $u\in \bC^n$ be the last column of a unitary matrix $g$ that satisfies equation \eqref{G2} for $\xi=I_s$ \eqref{52},
with $y$ in \eqref{51}.
Then  $u_k=0$ must be zero, except for $k= 2+s$ and $k=1,2$. Explicitly, we have
\be
\vert u_{2+s}\vert^2 = f_s(y)
\quad \hbox{with}\quad
f_s(y): = \frac{\sin((s+1) y)\sin((n-1)y)}{\sin(s y)\sin(n y)},
\quad \forall  s=1,2,\dots,  n-3.
\label{511}\ee
The functions $f_s$ satisfy $0< f_s(y) < 1$, possibly except for a finite set of $y$ values.
\end{lem}
\begin{proof}
We begin with a general observation. Suppose  that $\delta_k(\xi)$ has multiplicity $1$ as an eigenvalue of
$\delta(\xi)$ and there exists an index $j$ such that $\delta_k(\xi) e^{2\ri y} = \delta_j(\xi)$.
By evaluation of the equality \eqref{E1} at $\zeta = e^{2\pi\ri  y} \delta_k(\xi)$, we  find that
$\vert u_k \vert^2=0$.

Defining  $Y:= e^{2\ri y}$,  $\delta(I_s)$ has the structure
\be
\delta(I_s) = \delta_1(I_s)\, \diag (1,1, Y,Y^2,\dots, Y^s, Y^{-m}, Y^{-m+1},\dots, Y^{-1}),
\qquad m: = n-2 - s.
\label{dIs}\ee
Therefore the above observation is applicable  except for $k=1,2, s+2$.
As for the formula \eqref{511}, it results by evaluation  of \eqref{E1} at $\zeta= e^{2\ri y} \delta_{2+s}(I_s)$.

Since  $u$ is a column of a unitary matrix and we have proved \eqref{511}, the functions $f_s$ must satisfy
$0\leq f_s(y) \leq 1$. It is easily seen that $f_s(y)$ is never zero in the interval \eqref{51}.
We believe that it is also never equal to $1$, but have not yet proved this in  full generality.
We note that $f_s(y)$ goes to $0$ through positive values as $y$ \eqref{51} tends to the boundary point $\pi/(n-1)$.
For $n=4$, $s=1$, the formula \eqref{511} can be converted into \eqref{F6} by trigonometric identities.
\end{proof}

The following theorem represents a generalization of Proposition \ref{prop42}.

\begin{thm}\label{thm54}
The fiber of the momentum map $\hat \beta$ is a copy of $S^3$ over each of the $n(n-3)$
vertices of $\cA_y$  given by the cyclic permutations of $I_s$ \eqref{52} for $s=1,\dots, (n-3)$.
Here, we suppose that $y$ satisfies \eqref{51} and $f_s(y) \neq 1$.
\end{thm}
\begin{proof}
Recall that, in general, the fiber over $\xi\in \cA_y$ has been identified with the quotient space $\SU(n)_{\delta(\xi)}/\SU(n)_{\xi}^{u_0}$.
By the cyclic symmetry, it is enough to consider the vertices $I_s$.
We have to choose a vector $u$ subject to the properties established in Lemma \ref{lemfs}.
In particular, $u$ satisfies
\be
0<  \sqrt{\vert u_1\vert^2 + \vert u_2 \vert^2} := r_s < 1
\quad \hbox{and}\quad \vert u_{2+s}\vert^2 = (1-r_s^2),
\label{512}\ee
and all its other components vanish.
Clearly,   we have the isotropy group
\be
\SU(n)_{\delta(I_s)} = \{\diag(\Gamma, \Gamma_3, \dots, \Gamma_n) \mid \Gamma \in \U(2),\,
\Gamma_j \in \U(1),\,\,j=3,\dots, n\} \simeq \U(2) \times \bT^{n-3},
\label{513}\ee
where we took into account  the determinant $1$ condition.
 Choosing $u_0:=u$ with non-zero components $u_2 = r_s$ and
$u_{2+s} = \sqrt{1-r_s^2}$, we  obtain
\be
\SU(n)_{I_s}^{u_0} = \{\diag(\gamma_1, \gamma_2,\dots, \gamma_n)\mid \gamma_j \in \U(1),\,  \gamma_2 = \gamma_{2+s}\} = \bT^{n-2}.
\label{514}\ee
A simple counting implies that the fiber
\be
\hat \beta^{-1}(\xi) \simeq \SU(n)_{\delta(\xi)}/\SU(n)_\xi^{u_0}
\label{515}\ee
is a compact manifold of dimension $3$. A further close inspection identifies it as $S^3 \simeq \SU(2)$.

Let us sketch the key step of the inspection for the specific case $s= (n-3)$.
In this case we have
\be
\delta(\xi) = \delta_1 \diag(1,1, Y, Y^2,\dots, Y^{n-3}, Y^{-1})\quad  \hbox{with}\quad  \xi:= I_{n-3},
\label{516}\ee
and
\be
\hat \mu(u_0, y) = Y \delta(\xi) + Q,
\label{517}\ee
where the entries of the matrix $Q$ are zero expect for $Q_{2,2}$, $Q_{n-1,n-1}$, $Q_{2,n-1}$ and $Q_{n-1,2}$.
Let $V_1,\dots, V_n$ denote the columns of $\hat A_0\in \SU(n)$, determined by
\be
\hat A_0 \delta(\xi) = \hat \mu(u_0, y) \hat A_0.
\label{518}\ee
The columns are eigenvectors of $\hat \mu(u_0, y)$ with the respective eigenvalues collected in $\delta(\xi)$.
Obviously, we can choose these eigenvectors as follows:
\be
V_1 = e_n,\, V_3= e_1,\, V_4 = e_3,\dots, V_{n-1}= e _{n-2},
\label{519}\ee
and $V_2 = x_1 e_2 + x_2 e_{n-1}$, $V_n = x_3 e_2 + x_4 e_{n-1}$   with some coefficients $x_i\in \bC$, where
$e_1, \dots, e_n$ denotes the standard basis of $\bC^n$. This generalizes the formula \eqref{F11},
and then the group action \eqref{act33}  is easily calculated to acquire  the following explicit form:
\be
\Gamma \mapsto \diag(\gamma_n, \gamma_2) \Gamma \diag(\gamma_1^{-1}, \gamma_2^{-1}),
\label{520}\ee
and
\be
\Gamma_3 \mapsto \gamma_1 \Gamma_3 \gamma_3^{-1},\,\, \Gamma_j \mapsto \gamma_{j-1} \Gamma_j \gamma_j^{-1}
\quad \hbox{for}\quad j=4,\dots, n.
\label{521}\ee
These `gauge transformations' allow us to identify the quotient space
$\hat \beta^{-1}(\xi) \simeq \SU(n)_{\delta(\xi)}/\SU(n)_\xi^{u_0}$
\eqref{id310} with the set of elements of the form
\be
\{\diag(\Gamma, 1,\dots, 1) \mid \Gamma \in \SU(2)\}.
\label{522}\ee
In fact, transforming the $\Gamma_j$ $(j=3,\dots, n)$ to $1$ fixes all components $\gamma_i$ up to an overall $\bZ_n$ factor
and $\Gamma$  \eqref{513} gets transformed into $\SU(2)$.
This is a natural generalization of the proof of Proposition \ref{prop42},
and essentially the same calculation works  for all singular vertices.
\end{proof}

\begin{rem}\label{remgenI}
Take an arbitrary convex linear combination
\be
\xi = \sum_{s=1}^{n-3} c_s I_s.
\label{523}\ee
It is clear that $\SU(n)_{\delta(\xi)}$ has the same structure as in \eqref{513}, and $\xi \in \partial \cA \cap \cA_y$.
The structure of the corresponding fiber depends on what components of $u$  are zero and whether
$(\vert u_1 \vert^2 + \vert u_2 \vert^2)$ is  zero or not.
 The investigation of this question and the generalization of the statements of Remark \ref{rem41} and Proposition \ref{prop44}
are left to the interested reader.
\end{rem}

\subsection{Additional results}\label{Sec52}

We have performed a numerical exploration  of the type (ii) polytopes up to $n=15$
using the   software packages  {\tt polymake} \cite{poly} and {\tt howzat} \cite{howzat}.  The results are summarized
in part by Figure \ref{Fig2} and in more detail by the list of observations and tables located in Appendix \ref{AppC}.
In the appendix,  not only the number of faces of different dimensions but also the coordinates of the vertices
are presented up to $n=7$.
We observe that the structure changes as we change $k$,  for $k \frac{\pi}{n} < y < (k+1)\frac{\pi}{n}$,
and usually changes also when $y$ passes through an excluded value without changing $k$.
We imposed the condition $0<y<\pi/2$, since the systems associated with $y$ and $(\pi-y)$  are isomorphic \cite{FK}.
\emph{We found $2n$ facets in all type (ii) cases considered, and conjecture that this always holds}.

The software programs that we used are  capable to  return  the vertices for specific numerical values of the
parameter $y$, after inputting the defining inequalities.
We combined this with a linear interpolation to determine the vertices in intervals of $y$ between two excluded values.
The resulting points can be checked to be indeed vertices, and can be utilized for finding the faces.
We next describe the justification of this procedure.

Let us suppose that
\be
 k \pi/n < y_1 < y_2 < (k+1)\pi/n
 \ee
  holds and $\xi^{(1)}$ and $\xi^{(2)}$ verify the inequalities \eqref{T24} at $y_1$ and $y_2$, respectively.
As a result,
\be
\xi^{(t)}:= t \xi^{(1)} + (1-t) \xi^{(2)}
\ee
verifies the inequalities at $y(t) = t y_1 + (1-t) y_2$ for any $0\leq t \leq 1$.
Moreover, if $\xi^{(1)}$ and $\xi^{(2)}$ satisfies the same subset of the defining inequalities \emph{as equalities},
then $\xi^{(t)}$ satisfies precisely the same equalities.
Suppose that the pertinent set of \emph{inequalities and equalities} admits at most a single solution for any $y$ in the relevant range.
In this case, if $\xi^{(1)}$ and $\xi^{(2)}$ are vertices, then $\xi^{(t)}$ is also a vertex.
This property is consistent with the  expectation (confirmed by all known examples) that the structure of $\cA_y$ remains constant
between two excluded values of the parameter $y$.

The following lemma is worth stating, since it can be used to establish the existence of $2n$ facets once one knows the existence and the
form of vertices in $(\partial \cA_y)^\reg$.

\begin{lem}\label{facetlemL}
Suppose that $R \in (\partial \cA_y)^\reg$ is a vertex in a type (ii) case with $\frac{k}{n} \pi < y < \frac{k+1}{n}\pi $ that satisfies
two equalities of the form
\be
\xi_\ell+\cdots + \xi_{\ell +k-1} =y
\quad\hbox{and}\quad
\xi_m + \cdots + \xi_{m+k} = y
\label{facetcondL}\ee
for some $\ell$ and $m$. Then all the $2n$ defining inequalities of the polytope are facet generating, i.e., $\cA_y$ has $2n$ facets.
\end{lem}
\begin{proof}
It was explained at the end of Section \ref{Sec32}
that the defining hyperplanes of the polytope $\cA_y$ that meet at a vertex $R\in \cA^\reg$ are facet generating.
Thus, our assumptions   guarantee that the hyperplanes \eqref{facetcondL} correspond to facets.
By invoking the cyclic permutation symmetry, we conclude that all the $2n$ defining hyperplanes are facet generating.
\end{proof}

\begin{rem}
  Lemma \ref{facetlemL} is applicable to all examples presented in Appendix \ref{AppC}.
It can also be used to give an alternative proof of the claim of Theorem \ref{thm51} about the $2n$ facets, but
we preferred to develop a direct proof there.
\end{rem}

Incidentally, one may also study the polytope $\cA_y$ at the so far excluded values of $y$.
For example,  we know that at $y = \pi \frac{k}{n}$ with $\gcd(k,n)=1$  it contains only  the point $\xi^*$ \eqref{xistar}.
Less severe degenerations occur at other excluded values, as exemplified by Figure \ref{Fig1}.

Finally, let us mention some results on the possible consecutive zeroes of $\xi\in \cA_y$.
This is important since $m$ consecutive zeroes of $\xi$ imply that $\SU(n)_{\delta(\xi)}$ contains
an $\U(m+1)$ factor, and this enters into the structure of the fiber $\hat \beta^{-1}(\xi)$.
It is plain that $m$ consecutive zeroes might appear only if $y$ belongs to the range  \eqref{T21}  with $k \geq m$.
We inspected the possibilities with $k=2$ up to $n=10$  via a direct `by hand' calculation.
In this range, we found that two consecutive zeroes occur (only) for $n=9$.

\begin{prop}\label{prop56}
Consider the type (ii) cases for $n=9$ and $k=2$ \eqref{T21}. Then, the points of $\cA_y$ with 2 consecutive zeroes are exhausted by
\be
\xi = (0,0,y, 0, \pi - 3y, 4y - \pi, \pi - 3y, 0,y)
\quad
\hbox{with}\quad
 \frac{1}{4} \pi  < y  <  \frac{1}{3} \pi,\quad y\neq \frac{2}{7}\pi,
\ee
and its cyclic permutations.
For $k=2$, $n\leq 10$ there is no other type (ii) case with some $\xi\in \cA_y$ having 2 consecutive zeroes.
\end{prop}
\begin{proof}
If $\xi_1=\xi_2=0$, then $\xi$ must have the form
\be
\xi=(0,0,y,0, \xi_5,\xi_6,\xi_7, 0,y).
\ee
In fact, this is a direct consequence of the inequalities imposed on the 2 term and 3 term sums.
Since $\xi_5 + \xi_6 \leq y$ and $\xi_4 + \xi_5 + \xi_6 \geq y$, we get $\xi_5 + \xi_6 =y$
and in the same way $\xi_6 + \xi_7 = y$. Thus, $\xi$ must satisfy
\be
\xi=(0,0,y,0, a,\pi -a,a, 0,y)
\label{xi9}\ee
with some $a$. Since the sum of the components of $\xi$ is $\pi$, we get $a= \pi - 3y$.
This  leads to $\xi_6= 4y - \pi$, which is positive for $y > \pi/4$.
We note that one passes to type (i) cases below the excluded value $\pi/4$,
and $2\pi/7$ is the only excluded value between $\pi/4$ and $\pi/3$.
 This follows from Proposition \ref{thmI1} since  $2/9$, $1/4$, $2/7$ and $1/3$
 are consecutive elements of $\fF_9$.

The interested reader can confirm the `no go' claim by a similar elementary calculation.
\end{proof}

We presented the above result together with the proof since it
can be useful for studying other cases, too.
One can check that $\xi$ \eqref{xi9} is a  vertex as it belongs to a one point face.
In fact, it satisfies 15 out of the 18 relevant inequalities as equalities.
The only exceptions are $\xi_4 + \xi_5 = \xi_7 + \xi_8 = a <y$ and $\xi_1 + \xi_2 = 0 < y$.

 It is an interesting problem to determine the structure of the quotient space
 $\SU(n)_{\delta(\xi)}/\SU(n)_\xi^{u_0}$ \eqref{id310}
for $\xi$ \eqref{xi9}, which gives the fiber $\hat \beta^{-1}(\xi)$ as a manifold.

\section{Summary and outlook}\label{Sec6}

In this paper we studied the fibers of the momentum map for the Liouville integrable systems
$(P(\mu_0(y)), \omega_{\mu_0(y)}, \hat \fH_b)$ of \cite{FK}
 for the type (ii) admissible values of the parameter $y$.
More exactly, we were interested in the fibers over the `singular boundary points' $\xi \in \partial \cA_y \cap \partial \cA$
of the momentum polytope $\cA_y$.
In Section \ref{Sec3}, we proved that \emph{all fibers are smooth isotropic submanifolds}.
The proof led to the following algorithm for
determining the structure of the fibers:

\begin{enumerate}

\item{Find the boundary of the polytope $\cA_y\subset \cA$ \eqref{T8} defined by the inequalities \eqref{T25}.}

\item{Take a point $\xi \in \partial \cA_y \cap \partial \cA$ and consider the isotropy group
$\SU(n)_{\delta(\xi)}$ with $\delta(\xi)$ \eqref{T9}.}

\item{Pick a unit vector $u_0\in \bC^n$ whose components satisfy equations \eqref{E1}, \eqref{E7}
and find a matrix $\hat A_0\in \SU(n)$ for which $\hat A_0 \delta(\xi) \hat A_0^{-1}= \hat \mu(y,u_0) \delta(\xi)$ holds,
using the definition \eqref{G1}}.

\item{Consider the subgroup $\SU(n)_\xi^{u_0} <\SU(n)_{\delta(\xi)}$ whose elements
have $u_0$ as an eigenvector.}

\item{Finally, establish the structure of the quotient space
$\SU(n)_{\delta(\xi)}/\SU(n)_\xi^{u_0}$  for the action of $\SU(n)_\xi^{u_0}$ on $\SU(n)_{\delta(\xi)}$
defined by  \eqref{act33}, \eqref{G23}.}

\end{enumerate}

\noindent
We demonstrated that the fiber $\hat \beta^{-1}(\xi)$ is diffeomorphic to
$\SU(n)_{\delta(\xi)}/\SU(n)_\xi^{u_0}$, which is a smooth manifold since the pertinent isotropy
groups are all equal to  the center $\bZ_n$ of $\SU(n)$.

In Sections \ref{Sec4} and \ref{Sec5}, we successfully applied the algorithm to find the fibers over the `singular
vertices' for an infinite
series of type (ii) cases.
Namely, for $\SU(n)$ with $n\geq 4$ and $y$ in \eqref{51}, we proved that the polytope $ \cA_y$ possesses $n(n-3)$ singular
vertices and all the corresponding fibers are diffeomorphic to the manifold $S^3\simeq \SU(2)$.

For possible future work, we remark
that the structure of the fiber $\hat \beta^{-1}(\xi)$ depends critically on two key features.
First, $m$ consecutive zeroes of $\xi$ imply $m+1$ consecutive equal eigenvalues of $\delta(\xi)$ by \eqref{T11}, which
give rise to a $\U(m+1)$ factor in the group $\SU(n)_{\delta(\xi)}$.
Second, in the case when the corresponding $m+1$ components of $u_0$ vanish, this factor is
 `cancelled' by a corresponding $\U(m+1)$  factor that appears in $\SU(n)_\xi^{u_0}$. Whether this happens or not
 can be established, in principle, by  inspection of the relation \eqref{E7}.
 If the residue in \eqref{E7} is non-zero, then the corresponding factor of $\SU(n)_\xi^{u_0}$
 is isomorphic to the group $\U(m)\times \U(1)$,  and one expects to obtain an interesting fiber.

The application of the algorithm to more general type (ii) cases is left as an open problem.
 The first example to be studied is the fiber over the vertex with double zeroes of $\xi$
exhibited in  Proposition \ref{prop56}. Based on a preliminary investigation,
\emph{in this case the fiber is expected to be a Lagrangian $\SU(3)$ submanifold.}
 We note in passing that the description of the singular fibers as quotient spaces by the action \eqref{act33}
appears reminiscent of the characterization of the Gelfand--Cetlin fibers as `balanced products'
given by Theorem 3.4.3 in \cite{CL}.

In addition to the open problems mentioned above and those noted in the main text,
there are several further questions that would deserve attention.
A major problem is to better understand  the global geometry of the
phase spaces $P(\mu_0(y))$ \eqref{I.3} and, if possible, connect them to known
symplectic  manifolds.
Another  question is whether our integrable systems can be related to
toric degenerations \cite{HK}, as is the case for other systems exhibiting spherical
singularities \cite{NNU,NU-pol}.

The knowledge of action variables is useful not only for analyzing the classical dynamics of integrable Hamiltonians, but also
for obtaining their semiclassical spectra.  For compact toric systems, fitting an integral lattice on the momentum polytope
leads directly to the spectra of the quantized action variables \cite{FKlim2,Ham}.
 The action variables for certain non-compact toric \cite{KL} relatives of our systems have been recently
investigated in \cite{LMRS} and it could be interesting to explore their quantization from the semiclassical viewpoint.
The quantization of the  type (ii)  systems considered in this paper is  an important open problem, especially
because of their interpretation as compactified Ruijsenaars--Schneider systems (see Appendix \ref{AppA}).
 For the type (i) cases,
 a quantization utilizing discrete difference operators built on an integral lattice over the `configuration polytope'
has been worked out   in detail \cite{vDV,GH},  giving not only the spectra but also the joint eigenvectors
of the commuting Hamiltonians.
However, as was noted in \cite{GH},  it is not clear if  this method can be adapted to type (ii) cases.
In these circumstances, a case study of the  quantization of simplest type (ii) systems is warranted,
and the explicit knowledge of the vertices  provided by our Theorem \ref{thm51}
is expected to be useful for such a study.

\medskip

\subsubsection*{Acknowledgements}
 We are grateful to
 T.F.~G\"orbe and T.~Waldhauser for explaning us the relevance of Farey neighbours, which led to the
 presented formulation of Proposition \ref{thmI1}.
LF also wishes to thank J.~Evans, M.~Fairon, F.~Fodor, K.-H. Neeb and L.~Polterovich
for discussions and correspondence.
He thanks the Sydney Mathematical Research Institute  as well for support
during a visit when this work was started.
LF was supported in part  by the
grant NKKP Advanced 152467 and
HRD was supported by ARC grants DP210100951 and DP210101102.

\medskip

\appendix
\section{Interpretation as compactified trigonometric RS systems}\label{AppA}

The double $P$ \eqref{I.1} carries \emph{two} natural Abelian Poisson algebras,
defined by the class functions of either $A$ or $B$ for $(A,B) \in P$.
These descend to two  sets of Hamiltonians in involution  on the reduced phase space $P(\mu_0(y))$, which
yield two \emph{isomorphic}\footnote{This  means that the
action variables of the two systems can be converted into each other by a symplectomorphism.}
Liouville integrable systems.
They can be interpreted as two copies of the so-called compactified trigonometric Ruijsenaars--Schneider (RS) system,
which is an integrable many-body model of `point particles' moving on the circle.
As recalled below, this interpretation is tied to describing the commuting Hamiltonians of one integrable system in terms of  action-angle coordinates
of the other one.
It should be noted that what follows is valid for any admissible parameter $y\in [0,\pi]$.

Before plunging into details, it is in order to outline the history of these integrable systems.
The system on ${\bC}P^{n-1}$ was originally constructed by Ruijsenaars \cite{RIMS95} via a direct compactification
of the local phase space of the Hamiltonian \eqref{A6}, which made its flow complete.
In \cite{RIMS95},  where the relevant system is called the $\mathrm{III}_b$ system,
a restriction equivalent to $0<y < \pi/n$ was imposed based on a heuristic argument.
Under this restriction, the reduction treatment using the double $P$ \eqref{I.1} was worked out in \cite{FKlim}, motivated also
by earlier reduction derivations (e.g.~\cite{FR}) of the complex holomorphic version of the trigonometric
RS system  \eqref{A7}.
The paper \cite{FKlim} contains an exhaustive analysis of the relation between the two equivalent systems carried by $P(\mu_0(y))$, too.
The general compact case, and in particular the realization of the intriguing dependence on $y\in [0,\pi]$, is due to \cite{FK}.

The integrable system stemming from the class functions of $B$ admits action-angle
coordinates on the dense open submanifold  $P(\mu_0(y))^\reg:= \hat \beta^{-1}(\cA_y^o)$ \eqref{G30}.
Such a coordinate system can be obtained \cite{FKlim,FK} with the aid of a
cross-section for the $\SU(n)_{\mu_0(y)}$ action on $\beta^{-1}(\cA_y^o) \cap \mu^{-1}(\mu_0(y))$.
To present this, first
define the $\SU(n)$-valued matrix function $\cL_y^\loc$ on $\cA_{y}^o \times \bT^{n-1}$
by
\be
\cL_y^\loc(\xi,\theta)_{j \ell} :=
\frac{\sin (ny)}{\sin (y)}
\frac{e^{\ri y} - e^{-\ri y}}{e^{\ri y} \delta_j(\xi) \delta_\ell(\xi)^{-1} - e^{-\ri y}}
v_j(\xi,y) v_\ell(\xi,\pi-y) \varrho(\theta)_\ell,
\label{A1}\ee
where $\theta = (\theta_1,\dots, \theta_{n-1})$ parametrizes the maximal torus of $\SU(n)$ by the formula \eqref{T14}
and $v_\ell(\xi,y) := \sqrt{z_\ell(\xi, y)}$ with
the functions $z_\ell$  \eqref{E5} that are positive on $\cA_y^o = \cA_{\pi-y}^o$.
Then,
for any $v\in \bR^n$ that has unit norm and last component $v_n\neq -1$,
define the real unitary matrix $g(v)$  by
\bea
&&g(v)_{jn} := - g(v)_{nj}:= v_j,
\quad
\forall j=1,\ldots, n-1,
\quad
g(v)_{nn} := v_n,\nonumber\\
&& g(v)_{jl} := \delta_{jl}- \frac{v_j v_l}{ 1 + v_n},
\quad
\forall j,l=1,\ldots, n-1.
\label{A2}\eea
Using the positive vector $v(\xi,y)$,  denote $g_y(\xi):= g(v(\xi,y))$.

It is proved in \cite{FKlim,FK} that the set
\be
\cS:= \left\{\left( g_y(\xi)^{-1} \cL_y^\loc(\xi,\theta) g_y(\xi), g_y(\xi)^{-1}
\delta(\xi)g_y(\xi) \right)\,\Big\vert \,
(\xi,\theta)\in \cA^o_{y} \times \bT^{n-1}\,\right\}
\subset \mu^{-1}(\mu_0(y))
\label{A3}\ee
represents a  global cross-section of the orbits of $\SU(n)_{\mu_0(y)}$
in $\beta^{-1}(\cA^o_{y}) \cap \mu^{-1}(\mu_0(y))$.
Thus, $\cS$ provides a model of the quotient space $P(\mu_0(y))^\reg\subset P(\mu_0(y))$.
This model comes equipped with action-angle coordinates since the symplectic form  $\omega_{\mu_0(y)}$ can be written as
$ \sum_{k=1}^{n-1} \mathrm{d} \theta_k \wedge \mathrm{d} \xi_k$ on $P(\mu_0(y))^\reg \simeq \cS$.

Consider the Liouville integrable system on $P(\mu_0(y))$ that arises from the
Abelian Poisson algebra $\pi_1^* C^\infty(\SU(n))$, i.e., from the class functions of $A$.
 After restriction to the dense open subset modeled by $\cS$ \eqref{A3}, the commuting Hamiltonians of this system become
 the functions of the form
\be
\cH(\xi, \theta) = h( \cL_y^\loc(\xi,\theta)), \qquad \forall h \in C^\infty(\SU(n))^{\SU(n)}.
\label{A4}\ee
The simplest globally smooth complex class function is the matrix trace, which gives
\be
 \tr \bigl(\cL_y^\loc\bigr)=  s \sum_{\ell =1}^n \varrho(\theta)_\ell
 \Biggl[ \prod_{\substack{j=1\\j \neq \ell}}^n \left[ 1 + 4 \frac{\sin^2 y}
{ \bigl((\delta_\ell/\delta_j)^{\frac{1}{2}} - (\delta_\ell/\delta_j)^{-\frac{1}{2}}\bigr)^2} \right] \Biggr]^{\frac{1}{2}}
\label{A5}\ee
with
$s= {\mathrm{sign}}\left(\frac{ \sin(y)}{\sin(ny)}\right)$.
By putting $\delta_j:= e^{2 \ri q_j}$ and $\varrho(\theta)_j:= e^{\ri p_j}$ and dropping the sign-factor $s$,
the real part of the trace yields the Hamiltonian
\be
H_y^\loc(q,p) = \sum_{\ell=1}^n \cos (p_\ell)
\Biggl[ \prod_{\substack{j=1\\j \neq \ell}}^n \left[ 1 -\frac{\sin^2 y}
{ \sin^2(q_\ell - q_j)} \right] \Biggr]^{\frac{1}{2}}.
\label{A6}\ee
This Hamiltonian  defines a dynamics of $n$ interacting points moving on the circle, subject to
the `center of mass conditions' $\prod_{j=1}^n \delta_j = \prod_{j=1}^n e^{\ri p_j} =1$.
It represents  a compact form of the standard trigonometric Ruijsenaars--Schneider Hamiltonian \cite{RS}:
\be
H_{\mathrm{trigo-RS}}(q,p) = \sum_{\ell=1}^n \cosh (p_\ell)
 \prod_{\substack{j=1\\j \neq \ell}}^n \left[ 1 +\frac{\sinh^2 y}
{ \sin^2(q_\ell - q_j)} \right]^{\frac{1}{2}}.
\label{A7}\ee

The flows of reduced Hamiltonians descending from the real and imaginary parts of $\tr(A^k)$ for $k=1,\dots, n-1$
can leave the dense open subset $P(\mu_0(y))^\reg$,
but are of course complete on the full phase space  $P(\mu_0(y))$. Action variables for this Liouville integrable system
are provided by the mapping $\hat \alpha$ induced by $\alpha(A,B):= \Xi(A)$, similarly to \eqref{T13}.
The $\hat \alpha$ image of $P(\mu_0(y))$ is again $\cA_y$,   and the corresponding fibers have the same structure that we described
in the main text for $\hat \beta$.
\emph{The message is that all results of the present paper can be interpreted as statements about the action map (alias momentum map)
for the compactified trigonometric RS system.}

It is an interesting open question to find the `particle positions' and complementary coordinates
at which the action variables of the compact RS system reach the vertices
of the momentum polytope.
Let $\eta$ be a vertex of $\cA_y$ and consider the fiber $\hat \alpha^{-1}(\eta)$.
In principle, this fiber can also be described as a subset of $\hat \beta^{-1}(Q_\eta)$,
where
\be
Q_\eta:= \hat \beta(\hat \alpha^{-1}(\eta)).
\label{A8}\ee
The question is to find $Q_\eta$, and the points of the $\hat \beta$-fibers over it that fill
$\hat \alpha^{-1}(\eta)$.
This question is open even for the `regular vertices' that correspond to fixed points for the torus action generated
by $\hat \alpha$, both for type (i) and type (ii) systems.

\section{Elaborations regarding Proposition \ref{thmI1}}\label{newAppB}

First,  let us explain the relation of Proposition \ref{thmI1} to  corresponding statements in \cite{FK}.

Fix $1 \leq \kappa \leq n-1$ coprime to $n$, and consider the conditions in \eqref{HardyI1}
for  integers $a,b,c,d$ supposing that
 $1 \leq b,d \leq (n-1)$.
It is a fundamental fact \cite{HW} that these requirements admit unique integer solutions, and the solutions also satisfy
\be
0 \leq  a,c \leq n-1, \quad  \kappa=a + c, \quad n = b+d \quad \hbox{and} \quad \gcd(a,b) = \gcd(c,d)=1.
\label{Hardy2}\ee
Note that $b=1$ and $a=0$ if $\kappa =1$, and one puts $\gcd(0,1):=1$.
It has been shown in \cite{FK} that all non-excluded values of $y/\pi$ in the intervals $(a/b, \kappa/n)$ and $(\kappa/n, c/d)$ are type (i)
and all the other non-excluded values of $y/\pi$ outside these intervals are type (ii).
The only property that has not been shown  (but conjectured) in \cite{FK} is that there are no excluded values inside
the intervals $(a/b, \kappa/n)$ or $(\kappa/n, c/d)$.
 If one knows about Farey sequences,
 this follows immediately since the equations \eqref{HardyI1} and \eqref{Hardy2} imply  that
$a/b$ and $c/d$ are the neighbours of $\kappa/n$ in $\fF_n$.
In fact \cite{HW,R}, two terms $p/q < r/s$ of $\fF_n$ (written in reduced form) are neighbours if and only if
$rq-ps=1$ and $q+s > n$.

Second, let us comment on the asymptotics of the number of type (ii) and type (i) intervals as functions of $n$.
  As is shown in \cite{HW} (Theorem 331), the leading
  asymptotics of the number of terms in $\fF_n$ is $3 n^2/\pi^2$.
   On the other hand the number of coprime  positive integers less than $n$ is Euler's totient function $\phi(n)$, and its asymptotics
   is linear in $n$. The maximum $n-1$ is attained for prime $n$.
(Note that in \cite{HW} $\fF_n$ is called Farey \emph{series.})

\section{Some technical details}\label{AppB}

In Section \ref{Sec3}, we proved that the fiber $\hat \beta^{-1}(\xi) \subset P(\mu_0(y))$ is an embedded submanifold,
diffeomorphic to the quotient space $\SU(n)_{\delta(\xi)}/ \SU(n)_{\delta(\xi)}^{u_0}$ with respect to the action \eqref{act33}.
The definition of the latter quotient space involved arbitrary choices, which
must be equivalent since the embedded submanifold structure of the fiber is uniquely determined.
For completeness, we here show explicitly how the various choices lead to equivalent models of the fiber.

There are two steps where one has a freedom of choice. First, after fixing $u_0$ one may replace
$\hat A_0$ in $\hat A_0 \delta(\xi) \hat A_0^{-1} = \hat \mu(u_0,y) \delta(\xi)$
 by $\hat A_0' = \hat A_0 T_0^{-1}$ using an arbitrary element $T_0 \in \SU(n)_{\delta(\xi)}$.
With the old choice, $X \in \SU(n)_{\delta(\xi)}^{u_0}$ acts on $T\in \SU(n)_{\delta(\xi)}$ by the
diffeomorphism $f_X$ given by
\be
f_X(T) = \hat A_0^{-1}  X \hat A_0 T X^{-1}.
\ee
This gets replaced by $f'_X$ that operates as
\be
f_X'(T) = (T_0  \hat A_0^{-1})  X (A_0  T_0^{-1})  T  X^{-1}.
\ee
Letting  $L_{T_0}: \SU(n)_{\delta(\xi)} \to \SU(n)_{\delta(\xi)}$ be the map $L_{T_0}(T) := T_0 T$,
we obtained the relation
\be
 f'_X = L_{T_0} \circ f_X \circ (L_{T_0})^{-1}, \qquad \forall X \in \SU(n)_{\delta(\xi)}^{u_0}.
 \ee
Thus, the two group actions are intertwined by the diffeomorphism $L_{T_0}$, which induces
a diffeomorphism between the quotient spaces.

The other independent freedom is that we can  replace $u_0$ by $u'_0 = \Gamma_0 Q_0 u_0$ for some $\Gamma_0 \in \U(1)$ and
$Q_0\in \SU(n)_{\delta(\xi)}$, and also replace $\hat A_0$ by $\hat A_0' = Q_0 A_0 Q_0^{-1}$.
Then, we have
\be
\SU(n)_{\delta(\xi)}^{u_0'} = Q_0 \SU(n)_{\delta(\xi)}^{u_0} Q_0^{-1},
\ee
and the corresponding actions on $\SU(n)_{\delta(\xi)}$ enjoy the identity
\be
f'_{Q_0 X Q_0^{-1}} = \Ad_{Q_0} \circ f_X \circ (\Ad_{Q_0})^{-1}, \qquad \forall  X \in \SU(n)_{\delta(\xi)}^{u_0},
\ee
where $\Ad_{Q_0}(T) := Q_0 T Q_0^{-1}$ for all $T\in \SU(n)_{\delta(\xi)}$.
Clearly, this relation also induces a diffeomorphism between the quotient spaces.

\section{Exploration of type $\mathrm{(ii)}$ cases up to $n=15$}\label{AppC}

We computed the type (ii) polytopes up to including $n=7$ with the software package {\tt  polymake}
 \cite{poly}
and for larger $n$ with the software package {\tt howzat} \cite{howzat}.

This led to the following observations:
\begin{itemize}
\item Every polytope has $2n$ facets.
\item Every polytope has exactly $n$ regular vertices, equivalently, a single vertex up to cyclic permutations.
\item The coordinate functions for the vertices (see below for some examples) are of the form $a + b x$
where $|b| < n$ and $|a| < \lfloor n/2 \rfloor$ for integers $a,b$.
\item For prime $n$ there are exactly $n-1$ intervals of type (i), all other intervals are of type (ii) and there
are $(|\fF_n|-1)/2 - n + 1$ such intervals where $\fF_n$ denotes the Farey sequence.
The polytopes corresponding to these intervals are all distinct. In particular this suggests that
the number of distinct type (ii) polytopes grows like $O(n^2)$.
\item For composite $n$ there are distinct intervals of type (ii) which have the same face vectors,  and those polytopes are probably isomorphic.
They occur for intervals in the Farey sequence  separated by fractions of the type
 $p/q$ where $q$ is a divisor of $n$ and $\gcd(p, q) = 1$.
(Such an example occurs for $n=6$ with $p/q=1/3$ and in Figure \ref{Fig2}  the relevant  face vector is written over two intervals.)
\item The number of facets, ridges, edges, and vertices are divisible by $n$ for odd $n$ and they are divisible by $n/2$ for even $n$.
 \item Double consecutive zeroes of a vertex $\xi$ exist for $n=9$ and $n\geq 11$, but no triple zeroes were found up to n=15.
\item The following table lists the number of intervals of type (i) and type (ii) for $n = 4, \dots, 15$,
and the minimal and maximal number of vertices (up to cyclic permutation) of all type (ii) polytopes for that $n$.
\[
\left(
\begin{array}{c|cccccccccccc}
n &  4 & 5 & 6 & 7 & 8 & 9 & 10 & 11 & 12 & 13 & 14 & 15 \\ \hline
\#(i) & 2 & 4 & 2 & 6 & 4 & 6 & 4 & 10 & 4 & 12 & 6 & 8 \\
\#(ii) & 1 & 1 & 4 & 3 & 7 & 8 & 12 & 11 & 19 & 17 & 26 & 28 \\
min\#v(ii) & 2 & 3 & 3 & 5 & 4 & 7 & 5 & 9 & 6 & 11 & 7 & 13 \\
max\#v(ii) & 2 & 3 & 4 & 6 & 10 & 15 & 21 & 35 & 56 & 84 & 126 & 210 \\
\end{array}
\right)
\]
\end{itemize}

The subsequent tables contain the results of calculations made by using  {\tt{polymake}}  \cite{poly}.
See also  Figure \ref{Fig2} and the description of the method outlined in Section \ref{Sec52}.
In the following we write $x = y/ \pi$. The regular vertex is always listed first and denoted by $R$.

\goodbreak
\smallskip
\noindent $\mathbf{n=4}$.
Type (ii):
$n=4, k=1, \tfrac13 < x < \tfrac12$, face vector $(8, 14, 8)$, $8=4+4$ facets, 4 facets are triangles and 4 facets are quadrangles. 14 edges; 8 vertices.
\[
\begin{array}{cccccc}
R: &-1+3x & 1-2x &  x & 1 - 2 x \\
I: & 0 &  x & 1 - 2 x & x
\end{array}
\]

\goodbreak
\noindent $\mathbf{n=5}$.
Type (ii): $n=5, k=1, \tfrac14 < x < \tfrac13$
face vector $(15, 35, 30, 10)$, $10 = 5+5$ facets,
5 facets with 8 vertices (heptahedra with 2 triangular faces and 5 quad faces),
5 facets with 6 vertices (pentahedra with 2 triangular and 3 quad faces (triangular prism)).
$30 = 10+20$ ridges, 10 triangles,  20 quadrangles. 35 edges; 15 vertices.
\[
\begin{array}{cccccc}
R: &-1 + 4 x & 1- 3x &  x &   x &  1 - 3 x \\
I: & 0 &  x &  1 - 3 x &  x &  x \\
J: & 0 &  x &  x &  1 - 3 x &  x
\end{array}
\]

\goodbreak
\noindent $\mathbf{n=6}$.
For $n=6$ there are 4 intervals of type (ii).

$n=6, k=1, \tfrac15 < x < \tfrac14$, face vector  $(24, 69, 80, 45, 12)$, 
$12 = 6+6$ facets, 6 facets with 8 vertices and 6 facets with 15 vertices.
\[
\begin{array}{ccccccc}
R: & -1 + 5 x & 1-4 x &  x &  x &   x &  1 - 4 x \\
I: &0 &  x &  1 - 4 x &  x &  x &  x \\
J: & 0 &  x &  x &  1 - 4 x &  x &  x  \\
K: & 0 &  x &  x &  x &  1 - 4 x &  x
\end{array}
\]

$n=6, k=1,  \tfrac14 < x < \tfrac13$, face vector $(18, 57, 74, 45, 12)$, 
$12 = 6+6$ facets, 6 facets with 6 vertices and 6 facets with 12 vertices.
\[
\begin{array}{ccccccc}
R: &  -1 + 4 x &  1 - 3 x &  x &  1 - 3 x &  -1 + 4 x &  1 - 3 x\\
I: &  0 &  x &  1 - 3 x &  -1 + 4 x &  1 - 3 x  & x \\
J: &  0  & x & 0 &  x &  1 - 3 x &  x
\end{array}
\]

$n=6, k=2,  \tfrac13 < x < \tfrac25$, 
 same face vector as in previous case with $k=1$, but the coordinates are different:
\[
\begin{array}{ccccccc}
R: &  -1 + 3 x &  1 - 2 x &  -1 + 3 x &  1 - 2 x &  -1 + 3 x & 2 - 5 x\\
I:  & 0 &  1 - 2 x &  -1 + 3 x &  1 - 2 x & -1 + 3 x &  1 - 2 x  \\
J:  &  0 &  x &  0 & 1 - 2 x & -1 + 3 x &  1 - 2 x
\end{array}
\]

$n=6, k=2,  \tfrac25 < x < \tfrac12$, face vector  $(24, 75,  92, 51, 12)$, 
12 facets with 12 vertices.
\[
\begin{array}{ccccccc}
R:  &  -2 + 5 x &  1 - 2 x &  1 - 2 x & -1 + 3 x &  1 - 2 x &  1 - 2 x \\
I: & 0  &  1 - 2 x   &  -1 + 3 x &  1 - 2 x &  1 - 2 x &  -1 + 3 x \\
J: & 0 &  -1 + 3 x &  1 - 2 x &  1 - 2 x &  -1 + 3 x &  1 - 2 x \\
K: &   0 & x & 0 &  1 - 2 x & -1 + 3 x &  1 - 2 x
\end{array}
\]

\goodbreak
\noindent $\mathbf{n=7}$.
For $n=7$ there are 3 intervals of type (ii).

$n=7, k=1, \tfrac16 < x < \tfrac15 $,  face vector $(35,119,175,140,63,14)$
\[
\begin{array}{cccccccc}
R: &  x & x & x & 1-5 x & 6 x-1 & 1-5 x & x \\
I:& 0 & x & 1-5 x & x & x & x & x \\
J:& 0 & x & x & 1-5 x & x & x & x \\
K:& 0 & x & x & x & 1-5 x & x & x \\
L:& 0 & x & x & x & x & 1-5 x & x
\end{array}
\]

$n=7, k=1, \tfrac15 < x < \tfrac14 $, face vector $(42, 161, 252, 196, 77, 14)$
\[
\begin{array}{cccccccc}
R:& x & 1-4 x & 5 x-1 & 1-4 x & 5 x-1 & 1-4 x & x \\
I:& 0 & x & 1-4 x & 5 x-1 & 1-4 x & x & x \\
J:& 0 & x & x & 1-4 x & 5 x-1 & 1-4 x & x \\
K:& 0 & x & 0 & x & x & 1-4 x & x \\
L:& 0 & x & 0 & x & 1-4 x & x & x \\
M:& 0 & x & x & 0 & x & 1-4 x & x
\end{array}
\]

$n=7, k=2, \tfrac13 < x < \tfrac25 $, face vector $(35, 133, 210, 168, 70, 14)$
\[
\begin{array}{cccccccc}
R:& 3 x-1 & 2-5 x & 3 x-1 & 3 x-1 & 2-5 x & 3 x-1 & 1-2 x \\
I:& 0 & 3 x-1 & 1-2 x & 3 x-1 & 2-5 x & 3 x-1 & 1-2 x \\
J:& 0 & 1-2 x & 3 x-1 & 2-5 x & 3 x-1 & 1-2 x & 3 x-1 \\
K:& 0 & 1-2 x & 3 x-1 & 0 & 1-2 x & 3 x-1 & 1-2 x \\
L:& 0 & 3 x-1 & 1-2 x & 0 & 1-2 x & 3 x-1 & 1-2 x
\end{array}
\]

\end{document}